\documentclass[11pt,a4paper]{article}
\usepackage{amssymb}
\usepackage{graphicx}
\usepackage{amsmath}
\usepackage{makeidx}
\usepackage{indentfirst}
\usepackage[T1]{fontenc}
\usepackage[utf8]{inputenc}
\usepackage{authblk}

\newcounter{resultnum}[section]
\newcounter{conclusionnum}[section]
\newcounter{conditionnum}[section]
\newcounter{conjecturenum}[section]
\newcounter{examplenum}[section]
\newcounter{exercisenum}[section]
\newtheorem{lemma}{Lemma}[section]

\newcounter{lemmanum}[section]
\newcounter{notationnum}[section]
\newtheorem{theorem}{Theorem}[section]

\newcounter{theoremnum}[section]
\newtheorem{definition}{Definition}[section]

\newcounter{definitionnum}[section]
\newtheorem{corollary}{Corollary}[section]

\newcounter{corollarynum}[section]
\newcounter{remarknum}[section]
\newtheorem{proposition}{Proposition}[section]

\newcounter{propositionnum}[section]
\newcounter{acknowledgementnum}[section]
\newcounter{algorithmnum}[section]
\newcounter{axiomnum}[section]
\newcounter{casenum}[section]
\newtheorem{claim}{Claim}[section]

\newcounter{claimnum}[section]
\newcounter{summarynum}[section]
\newcounter{problemnum}[section]
\newenvironment{proof}[1][]{\textbf{Proof.} }{}

\begin{document}

\title{The Algebraic Index Theorem and\\
Deformation Quantization of Lagrange--Finsler\\
and Einstein Spaces}
\date{July 1, 2013}

\author{\textbf{Sergiu I. Vacaru}\thanks{%
sergiu.vacaru@uaic.ro}}

\affil{\small {\qquad \quad } Rector's Office, Alexandru Ioan Cuza
University, \newline Alexandru Lapu\c sneanu street, nr. 14, UAIC -- Corpus
R, office 323;\newline  Ia\c si,\ Romania, 700057 }%
\renewcommand\Authands{
and }

\maketitle

\begin{abstract}
Various types of Lagrange and Finsler geometries, Einstein gravity, and modifications, can be modelled by nonholonomic distributions on tangent bundles/ manifolds when the fundamental geometric objects are adapted to nonlinear connection structures. We can convert such geometries and physical theories into almost K\"{a}hler/ Poisson structures on (co)tangent bundles. This allows us to apply deformation quantization formalism to almost symplectic connections induced by Lagrange--Finsler and/or Einstein fundamental geometric objects. There are constructed respective nonholonomic versions of the trace density maps for the zeroth Hochschild homology of deformation quantization of distinguished algebras (in this work, adapted to nonlinear connection structure). Our main result consists in an algebraic index theorem for Lagrange--Finsler and Einstein spaces. Finally, we show how the Einstein field equations for gravity theories and geometric mechanics models can be imbedded into the formalism of deformation quantization and index theorem.


53D55, 58J22, 81S10, 83C45, 83C99
\end{abstract}



\section{Introduction}

The goal of this paper is to prove an algebraic index theorem for generalized Finsler and/or Einstein spaces and show how corresponding gravitational field equations, and their solutions, can be encoded into a
nonholonomic version of Fedosov manifolds for deformation quantization. It
is a partner work of \cite{vlfedq,avqgr} and belongs to a series of our
articles on geometric methods in deformation, A--brane, bi--connection etc
quantization models of gravity and gauge theories and geometric mechanics %
\cite{vqgr2,vegpla,vqgrbr,vggr,esv}. The main motivations
for this kind of investigations come from classical and quantum nonlinear
fundamental physical equations encoded as geometric structures on
nonholonomic manifolds \cite{fn01},  and there is an important physical task to quantize
such generic nonlinear theories. Similar constructions arise also as
mathematical problems for developing quantum/ noncommutative versions of
Riemann--Finsler and Hamilton--Lagrange spaces and in relation to possible
applications of geometric methods in modern particle physics.

The concept of nonholonomic manifold came from geometric mechanics and
classical and quantum theories with non--integrable (equivalently,
nonholonomic, or anholonomic) constraints \cite{fn02}. In our works, we
follow certain methods and formalism developed in the geometry of classical
and quantum nonholonomic manifolds, Lagrange -- Finsler/ Hamilton -- Cartan
geometries and applications \cite{vr1,vr1a,vr2}, see comprehensive reviews
and references in \cite{bejf,vrflg}. It allows us to
elaborate an unified geometric approach for the above mentioned types of
classical and quantum models, working, for simplicity, with nonholonomic
distributions defining nonlinear connection (N--connection) structures via
non--integrable splitting into Whitney sums of corresponding tangent spaces
(see next section for a summary of necessary definitions and results).

It is well--known that the algebra of pseudo--differential operators on a
compact manifold $M$ can be viewed as a quantum model (i.e. quantization) of
the cotangent bundle $T^{\ast }M.$ In this framework, the Atiyah--Singer
index theorem \cite{atiyah} relates the index of an elliptic
pseudo--differential operator on $M$ to the Todd class of $M,$ when the
Chern character of the bundle is associated naturally with the symbol of the
pseudo--differential operator under consideration. B. Fedosov developed a
deformation quantization (via formal power series with complex coefficients)
of an arbitrary symplectic manifold $\mathcal{M}$ \cite{fed1,fed2}. A
natural analogue of index theorem was proposed in K--theory working with the
quantum algebra of functions on symplectic manifolds. Here we note that quantum deformations use the correspondence between physical structures and deformations of algebraic noncommutative structures.

However, the bulk of physical theories usually are not encoded in terms of
the geometry of "pure" symplectic manifolds or further developments as
Poisson manifolds and their deformation quantization \cite{konts1,konts2}. A
quite universal geometric scheme including (semi) Riemannian \cite{fn03} and Finsler spaces, Lagrange
and Hamilton mechanics and generalized nonholonomic Einstein spa\-ces  can
be elaborated in the framework of the geometry of almost K\"{a}hler
manifolds \cite{vlfedq,vrflg,vqgrbr} with associated canonical N--connection
structure. This scheme can be generalized for additional geometric/physical
structures which implies new important results in geometric and/or
deformation quantization. For instance, we cite some applications of
Fedosov's quantization in modern particle physics (see, for instance, \cite%
{castro1,grl}) and for almost symplectic geometry, in general, with
nontrivial torsion, see \cite{karabeg1,karabeg,karabeg2}.

In order to elaborate explicit models of deformation quantization for
generally constrained physical systems, the most important task was to prove
that the fundamental geometric objects defining such almost symplectic
spaces are determined naturally, and in a unique form, by generating
(fundamental) Finsler/Lagrange functions, generic off--diagonal metrics%
 (
in particular, being solutions of the Einstein equations and/or
generalizations). Such a project was derived from Fedosov works \cite%
{fed1,fed2} and Karabegov--Schlichenmaier developments \cite{karabeg1} for
deformation quantization of almost K\"{a}hler geometries and realized in a
series of our works \cite{vqgr2,vegpla,vlfedq,avqgr}.

Generalizations of the famous Atiyah--Singer index theorem for nonholonomic
Clifford bundles, in particular, generated for Lagrange--Finsler spaces, and
gerbes were studied in Refs. \cite{vgonz}. Nevertheless, those
versions do not provide straightforward relations to almost K\"{a}hler
models of Finsler/--Lagrange and/or Cartan/--Hamilton geometries, Einstein
manifolds and their deformation quantization. In this article, we provide a
local version and a simple proof of the algebraic index theorem for the
mentioned types of nonholonomic almost symplectic manifolds. We follow the
scheme using an explicit formula for the trace
density map from the quantum algebra of functions on an arbitrary symplectic
manifold $\mathcal{M}$ to the top degree cohomology of $\mathcal{M}$ \cite%
{feigin}. In our case, that formalism is adapted to N--connection
structures. More precisely, our constructions are built from fundamental
geometric objects derived from canonical almost symplectic forms and
nonlinear and linear connections for almost K\"{a}hler--Finsler manifolds
(and various modifications for Lagrange--Hamilton geometric mechanics and/or
the Einstein gravity theory).

Throughout the paper we assume the summation over repeated right indices and use left indices as abstract labels for certain geometric objects; boldface symbols being considered for spaces/objects which are enabled/ad\-apted to nonlinear connection structures (such a system of notations was elaborated in Ref. \cite{vrflg}).

The paper is organized as follows. In section \ref{preliminaries}, we
summarize the necessary results on almost K\"{a}hler geometric models of
Lagrange--Finsler and Einstein spaces. Section \ref{snfq} is devoted to
Fedosov quantization of Einstein--Finsler spaces and the formalism of trace
density maps adapted to nonlinear connection structure. We prove the Main
Result (a local  Atiyah--Singer index theorem for
Lagrange--Finsler and Einstein Spaces) in section \ref{sasindth}. The
problem of algebraic index encoding of Einstein equations and exact
solutions in gravity theories is also discussed.

\vskip5pt

\textbf{Acknowledgement: } The work is partially supported
by the Program IDEI, PN-II-ID-PCE-2011-3-0256. The author is grateful to organizers and participants of "Applied Seminar" at Department of Mathematics, University College of London, UK; organizers of the Conference 3 Quantum: Algebra, Geometry, Information (Tallinn University of Technology, Estonia; July 10-14, 2012) and XXXI Workshop on Geometric Methods in Physics, Bialowieza, Poland, June 24-30, 2012, where the results of  this paper where presented in memory of Prof. B. Fedosov. The author thanks the Referee for important requests and suggestions which improved substantially the content of the article.

\section{ Almost K\"{a}hler Models of Lagrange--Finsler \& Einstein Spa\-ces}

\label{preliminaries} In this  section, we recall some
necessary results on encoding data for nonholonomic manifolds/bundles as
almost K\"{a}hler spaces with fundamental geometric objects adapted to
nonlinear connection (N--connection) structure.

\subsection{Nonholonomic distributions with associated N--connecti\-on}

Let us consider a $(n+m)$--dimensional nonholonomic manifold (see definition
in note \cite{fn1}) $\mathbf{V,}$ \ $\dim \mathbf{V\geq }2+1,$ of
necessary smooth class. Such a space is enabled with a conventional $n+m$
splitting when local coordinates $u=(x,y)$ on an open region $U\subset
\mathbf{V}$ are labeled in the form $u^{\alpha }=(x^{i},y^{a}),$ $\alpha =(i,a)$,where
indices $i,j,k,...=1,2,$ $...,n$ and $a,b,c...=n+1,...,n+m.$
There are changes of coordinates $(x^{i},y^{a})$ $\rightarrow (\tilde{x}^{i},%
\tilde{y}^{a})$ when \ $\tilde{x}^{i}$ are functions only of $x^{i}.$ It
results that $\frac{\partial }{\partial y^{a}}=\frac{\partial \tilde{y}^{b}}{%
\partial y^{a}}\frac{\partial }{\partial y^{a}} $ and so $\frac{\partial }{%
\partial y^{a}}$ locally span an integrable distribution on $\mathbf{V.}$ We
may fix a supplement of it locally spanned by $\mathbf{e}_{i}=\delta _{i}=%
\frac{\partial }{\partial x^{i}}-N_{i}^{a}(u)\frac{\partial }{\partial y^{a}}%
.$ This is an example of non--integrable distribution $\mathcal{N}$\ \
transforming $\mathbf{V}$ into a nonholonomic manifold. If the functions $%
N_{i}^{a}(u)$ are chosen in such a way that $\delta _{j}=\frac{\partial
\tilde{x}^{i}}{\partial x^{j}}\delta _{i},$ one  says that $\mathbf{N=}%
\{N_{i}^{a}(u)\}$ defines a nonlinear connection structure. In particular, $%
\mathbf{V}$ can be the total space of submersion over a $n$--dimensional
manifold $M.$ If $\mathbf{V}=TM,$ where $TM$ is the total space of a tangent
bundle to a $n$--dimensional manifold $M,$ the N--connection structure can
be considered any one introduced for a (pseudo) Finsler, or Lagrange,
geometry modeled on such a tangent bundle \cite{vrflg}.  For $vTM$ being
the vertical distributions on $TM,$ we introduce:

\begin{definition}
\label{defnc}Any Whitney sum
\begin{equation}
TTM=hTM\oplus vTM  \label{whitney}
\end{equation}%
defines a nonlinear connection (N--connection) structure paramet\-riz\-ed by
local vector fields $e_{i}=\frac{\partial }{\partial x^{i}}-N_{i}^{a}(x,y)%
\frac{\partial }{\partial y^{a}}$ on $TM.$
\end{definition}

A N--connection states on $TM$ a conventional horizontal (h) and vertical
(v) splitting (decomposition). If a h--v splitting, $T\mathbf{V}=h\mathbf{V}%
\oplus v\mathbf{V},$ exists on a general nonholonomic manifold $\mathbf{V,}$
we call such a space  N--anholonomic. For gravity theories, we shall
consider $\mathbf{V}$ to be a (pseudo) Riemannian spacetime.

Let $L(x,y)$ be a regular differentiable Lagrangian on $U\subset TM$ (for
analogous models $U\subset \mathbf{V})$ with non-degenerate Hessian
(equivalently, fundamental tensor field)
\begin{equation}
g_{ab}(x,y)=\frac{\partial ^{2}L}{\partial y^{a}\partial y^{b}}.  \label{lm}
\end{equation}

\begin{definition}
\label{ldls} (\cite{kern,vrflg}) A Lagrange space $L^{n}=(M,L(x,y))$ is
defined by a function $TM\ni (x,y)\rightarrow L(x,y)\in \mathbb{R},$ i.e. a
fundamental Lagrange function, which is differentiable on $\widetilde{TM}%
:=TM\backslash \{0\},$ where $\{0\}$ is the set of null sections, and
continuous on the null section of $\ \pi :TM\rightarrow M$ and such that the
(Hessian) tensor field $g_{ab}(x,y)$, (\ref{lm}), is non-degenerate and of
constant signature on $\widetilde{TM}.$
\end{definition}

A (pseudo) Lagrange space can be effectively modelled on a nonholonomic
(pseudo) Riemann manifold of even dimension $\mathbf{V}=(\mathbf{V}^{2n},%
\underline{\mathbf{g}}), \dim \mathbf{V}^{2n}$ $=2n,n\geq 2,$ prescribing a
generating function $\mathcal{L}(x,y),$ with $u=(x,y)\in \mathbf{V}^{2n},$
satisfying the conditions of Definition \ref{ldls}. We shall use two symbols
$L$ and $\mathcal{L}$ in order to distinguish what type of geometric
mechanical model (a classical one, from Lagrange geometry, or an analogous,
pseudo--Lagrange one, on a pseudo--Riemannian manifold) we involve in our
constructions.

Let us consider a regular curve $u(\tau )$ with real parameter $\tau ,$ when $%
u:\tau \in \lbrack 0,1]\rightarrow x^{i}(\tau )\subset U.$ It can be lifted
to $\pi ^{-1}(U)\subset \widetilde{TM}$ as $\widetilde{u}(\tau ):\tau \in
\lbrack 0,1]\rightarrow \left( x^{i}(\tau ),y^{i}(\tau )=\frac{dx^{i}}{d\tau
}\right) $ since the vector field $\frac{dx^{i}}{d\tau }$ does not vanish on
$\widetilde{TM}.$ For a different calculus \cite{vrflg}, we prove:

\begin{theorem}
\label{th1}The Euler--Lagrange equations, $\frac{d}{d\tau }\frac{\partial L}{%
\partial y^{i}}-\frac{\partial L}{\partial x^{i}}=0,$ are equivalent to the
nonlinear geodesic (semi--spray) equations%
\begin{equation}
\frac{d^{2}x^{i}}{d\tau ^{2}}+2G^{i}(x,y)=0,  \label{ngeq}
\end{equation}%
where $G^{i}=\frac{1}{2}g^{ij}\left( \frac{\partial ^{2}L}{\partial
y^{j}\partial x^{k}}y^{k}-\frac{\partial L}{\partial x^{j}}\right),$ for $%
g^{ij}$ being the inverse to $g_{ij}$ (\ref{lm}).
\end{theorem}

The conditions, and proof, of this theorem can be redefined on arbitrary $%
\left( \mathbf{V}^{2n},\underline{\mathbf{g}}\right) $ endowed with a
nonholonomic distribution induced by any prescribed $\mathcal{L}$ subjected
to conditions similar to those for a regular Lagrangian in mechanics.

\begin{proposition}
\label{pnaf}There are canonical frame and co--frame structures, $\mathbf{e}%
_{\alpha }$ and $\mathbf{e}^{\alpha },$ respectively defined by the
canonical N--connection
\begin{equation}
\ \ ^{c}N_{i}^{a}:= \frac{\partial G^{a}}{\partial y^{i}},
\label{clnc}
\end{equation}%
when \cite{fn04}
\begin{eqnarray}
 \mathbf{e}_{\alpha }=(\mathbf{e}_{i} &=&\frac{\partial }{\partial x^{i}}%
-N_{i}^{a}\frac{\partial }{\partial y^{a}},e_{b}=\frac{\partial }{\partial
y^{b}}),  \label{dder} \\
\mathbf{e}^{\alpha }=(e^{i} &=&dx^{i},\mathbf{e}^{b}=dy^{b}+N_{i}^{b}dx^{i}),
\label{ddif}
\end{eqnarray}%
for $\mathbf{e}_{\alpha }\rfloor \mathbf{e}^{\beta }=\delta _{\alpha
}^{\beta },$ where by $\rfloor $ we note the interior products and $\delta
_{\alpha }^{\beta }$ being the Kronecker delta symbol. Such N--elongated
partial derivative/ differential operators can be defined for any sets $%
N_{i}^{a}$ which are not obligatory represented in a canonical form $\ \
^{c}N_{i}^{a}).$
\end{proposition}

\begin{proof}
The results can be proven on any N--anholnomic manifold $\mathbf{V}$ by
explicit constructions using formula (\ref{clnc}) with $G^{a}$ determined in
(\ref{ngeq}); it is imposed the condition that such (co) frames should
depend linearly on coefficients of respective N--connections. $\square $
\end{proof}

\begin{definition}
The N--lift of the fundamental tensor fields $g_{ab}$ (\ref{lm}) on $TM$ \
(in general, from any $h\mathbf{V}$ to $\mathbf{V}$) is a Sasaki type metric
(distinguished metric, d--metric)
\begin{equation}
\mathbf{g}=\mathbf{g}_{\alpha \beta }\ \mathbf{e}^{\alpha }\otimes \mathbf{e}%
^{\beta }=g_{ij}(x,y)e^{i}\otimes e^{j}+g_{ab}(x,y)\mathbf{e}^{a}\otimes
\mathbf{e}^{b},  \label{sasakmetr}
\end{equation}%
where $g_{ij}$ is stated by $g_{ab}$ following $g_{ij}=g_{n+i\ n+j}.$
\end{definition}

Canonical N--connection and d--metric structures can be constructed on any
(pseudo) Riemannian manifold of even dimension, $\mathbf{V}=\left(\mathbf{V}%
^{2n},\underline{\mathbf{g}}\right),\dim$ $\mathbf{V}^{2n}=2n,n\geq 2,$ with
given metric structure $\underline{\mathbf{g}}=\mathbf{g}_{\underline{\alpha
}\underline{\beta }}\ \mathbf{\partial }^{\underline{\alpha }}\otimes
\mathbf{\partial }^{\underline{\beta }}$ if a generating Lagrange function $%
\mathcal{L}(x,y)$ is correspondingly prescribed on $\mathbf{V}^{2n}.$ Any
metric $\underline{\mathbf{g}}$ can be represented in a Lagrange--Sasaki
form (\ref{sasakmetr}) via frame transforms
\begin{equation}
\mathbf{e}_{\alpha }=\mathbf{e}_{\alpha }^{\ \underline{\alpha }}\mathbf{%
\partial }_{\underline{\alpha }}\mbox{\ and \ }\mathbf{e}^{\alpha }=\mathbf{e%
}_{\ \underline{\alpha }}^{\alpha }\mathbf{\partial }^{\underline{\alpha }%
};\ \mathbf{g}_{\alpha \beta }=\mathbf{e}_{\alpha }^{\ \underline{\alpha }}%
\mathbf{e}_{\beta }^{\ \underline{\beta }}\mathbf{g}_{\underline{\alpha }%
\underline{\beta }},  \label{fralgeq}
\end{equation}%
where matrices $\mathbf{e}_{\alpha }^{\ \underline{\alpha }}$ and $\mathbf{e}%
_{\ \underline{\alpha }}^{\alpha }$ can be chosen to be mutually inverse.
The explicit formulas depend on the type of nonholonomic structure we
prescribe by $\mathcal{L}(x,y)$ (via induced N--connection (\ref{clnc}) and
N--adapted frames (\ref{dder}) and (\ref{ddif})). For instance, if $\
\mathbf{g}_{\underline{\alpha }\underline{\beta }}$ is given as a solution
of the Einstein equations in four dimensional (4--d) general relativity and $%
\mathbf{g}_{\alpha \beta }$ is determined by a chosen $\mathcal{L},$ we can
always define certain $\mathbf{e}_{\alpha }^{\ \underline{\alpha }}$
encoding the  gravitational data into analogous mechanical ones, and/or
inversely.
In 4--d gravity, there are six independent components of \ $\ \mathbf{g}_{%
\underline{\alpha }\underline{\beta }}$ (from ten ones for a symmetric
second rank tensor, we can always fix four ones by corresponding coordinate
transforms). The values $\mathbf{e}_{\alpha }^{\ \underline{\alpha }}$ can
be defined as some solutions of algebraic  equations (\ref{fralgeq}) for given coefficients of metrics. Such constructions allow us always to
introduce (pseudo) Lagrange variables on a (pseudo) Riemannian manifold and,
inversely, any regular Lagrange mechanics can be geometrized as a Riemannian
space enabled with additional nonholonomic structure determined by $L(x,y).$

We have the following:

\begin{proposition}
A canonical N--connection $\mathbf{N}$ (\ref{clnc}) defines a canonical
almost complex structure $\mathbf{J.}$
\end{proposition}

\begin{proof}
The linear operator $\mathbf{J}$ acting on $\mathbf{e}_{\alpha }=(\mathbf{e}%
_{i},e_{b})$ (\ref{dder}) is defined by
\begin{equation*}
\mathbf{J}(\mathbf{e}_{i})=-\mathbf{e}_{n+i}\mbox{\ and \
}\mathbf{J}(e_{n+i})=\mathbf{e}_{i}.
\end{equation*}

This is a global almost complex structure ($\mathbf{J\circ J=-I}$ for $%
\mathbf{I}$ being the unity matrix) on $TM$ completely determined by $L(x,y).
$ $\square $
\end{proof}

\begin{definition}
The Neijenhuis tensor field for an almost complex structure $\ \mathbf{J}$
determined by a N--connection (i.e. the curvature of N--connection) is
\begin{equation*}
\ \ ^{\mathbf{J}}\mathbf{\Omega (X,Y):}=\mathbf{\
-[X,Y]+[JX,JY]-J[JX,Y]-J[X,JY],}  \label{neijt}
\end{equation*}%
for any vectors $\mathbf{X}$ and $\mathbf{Y.}$
\end{definition}

The N--adapted (co) bases (\ref{dder}) and (\ref{ddif}) are nonholonomic
when $\lbrack \mathbf{e}_{\alpha },\mathbf{e}_{\beta }]=\mathbf{e}_{\alpha }%
\mathbf{e}_{\beta }-\mathbf{e}_{\beta }\mathbf{e}_{\alpha }=W_{\alpha \beta
}^{\gamma }\mathbf{e}_{\gamma }.$ The nontrivial (antisymmetric) anholonomy
coefficients are $W_{ia}^{b}=\partial _{a}N_{i}^{b}$ and $W_{ji}^{a}=\Omega
_{ij}^{a},$ with the coefficients of N--connection curvature computed $%
\Omega _{ij}^{a}=\frac{\partial N_{i}^{a}}{\partial x^{j}}-\frac{\partial
N_{j}^{a}}{\partial x^{i}}+N_{i}^{b}\frac{\partial N_{j}^{a}}{\partial p_{b}}%
-N_{j}^{b}\frac{\partial N_{i}^{a}}{\partial p_{b}}.$

We can introduce on a (pseudo) Riemannian manifold an analogous (pseudo)
Finsler structure defined for any $\mathcal{L=F}^{2}(x,y),$ where an
effective Finsler metric $\mathcal{F}$ is a differentiable function of class
$C^{\infty }$ in any point $(x,y)$ with $y\neq 0$ and is continuous in any
point $(x,0);$ $\mathcal{F}(x,y)>0$ if $y\neq 0;$ it satisfies the
homogeneity condition $\mathcal{F}(x,\beta y)=|\beta |\mathcal{F}(x,y)$ for
any nonzero $\beta \in \mathbb{R}$ and the Hessian (\ref{lm}) computed for $%
\mathcal{L=F}^{2}$ is positive definite. A nonholonomic manifold can be
alternatively modelled equivalently as an analogous Finsler space. On
convenience, in this work, we shall consider both types of alternative
modelling of gravity theories (with nonhomogeneous configurations, i.e.
Lagrange type, and homogeneous ones, i.e. Finsler type).

\subsection{N--adapted almost K\"{a}hler structures}

An almost K\"{a}hler geometry can be also adapted to (induced by) canonical
N--connections.

\begin{definition}
An almost symplectic structure on $\mathbf{V}$ is defined by a nondegenerate
2--form $\theta =\frac{1}{2}\theta _{\alpha \beta }(u)e^{\alpha }\wedge
e^{\beta }.$
\end{definition}

An almost Hermitian model of a nonholonomic (pseudo) Riemannian spa\-ce $%
\mathbf{V}^{2n},\dim \mathbf{V}^{2n}=2n,n\geq 1,$ equipped with an
N--connection structure $\mathbf{N}$ is defined by a triple $\mathbf{H}%
^{2n}=(\mathbf{V}^{2n},\theta ,\mathbf{J}),$ where $\mathbf{\theta (X,Y)}%
\doteqdot \mathbf{g}\left( \mathbf{JX,Y}\right) .$ In addition, we have that
a space $\mathbf{H}^{2n}$ is almost K\"{a}hler, denoted $\mathbf{K}^{2n},$
if and only if $d\mathbf{\theta }=0.$ In this paper, we consider that  a real manifold is almost K\"{a}hler if it is endowed with a closed almost symplectic 2--form $\theta$.

We recall that for pseudo--Lagrange/ Finsler modelling of Einstein gravity %
\cite{vegpla,vlfedq,vrflg} (see also discussion and references therein; for
Finsler spaces, the original result is due to \cite{matsumoto}):

\begin{theorem}
\label{thmr2}Having chosen a generating function $\mathcal{L}(x,y)$ (or $%
\mathcal{F}(x,y))$ on a (pseudo) Riemannian manifold $\mathbf{V}^{2n},$ we
can model this space as an almost K\"{a}hler geometry, i.e. $\ ^{\mathcal{L}}%
\mathbf{H}^{2n}=\ ^{\mathcal{L}}\mathbf{K}^{2n},$ where the left labels
emphasize that such structures are induced nonholonomically by $\mathcal{L}$
(or $\mathcal{F}).$
\end{theorem}

Let us consider a metric $\mathbf{g}$ (\ref{sasakmetr}) and some structures $%
\mathbf{N}$ and $\mathbf{J}$ canonically defined by a prescribed $\mathcal{L}%
.$ We define $\mathbf{\theta (X,Y)}\doteqdot \mathbf{g}\left( \mathbf{JX,Y}%
\right) $ for any vectors $\mathbf{X}$ and $\mathbf{Y}$ and compute locally
\begin{eqnarray}
\mathbf{\theta } &=&\frac{1}{2}\theta _{\alpha \beta }(u)e^{\alpha }\wedge
e^{\beta }=\frac{1}{2}\theta _{\underline{\alpha }\underline{\beta }}(u)du^{%
\underline{\alpha }}\wedge du^{\underline{\beta }}  \label{sform} \\
&=&g_{ij}(x,y)e^{n+i}\wedge
dx^{j}=g_{ij}(x,y)(dy^{n+i}+N_{k}^{n+i}dx^{k})\wedge dx^{j}.  \notag
\end{eqnarray}%
Introducing the the form $\omega =\frac{1}{2}\frac{\partial \mathcal{L}}{%
\partial y^{n+i}}dx^{i},$ we get $\mathbf{\theta }=d\omega ,$ i.e. $d\mathbf{%
\theta }=dd\omega =0.$ We conclude that using a generating function $%
\mathcal{L}$ (or $\mathcal{F}),$ via canonical $\mathbf{g,N}$ and $\mathbf{J}
$, a (pseudo) Riemannian/Finsler/Lagrange space can be represented
equivalently as an almost K\"{a}hler geometry.

\begin{definition}
A linear connection on $\mathbf{V}^{2n}$ is a distinguished connection
(d--connection) $\mathbf{D}=(hD;vD)=\{\mathbf{\Gamma }_{\beta \gamma
}^{\alpha }=(L_{jk}^{i},\ ^{v}L_{bk}^{a};C_{jc}^{i},\ ^{v}C_{bc}^{a})\},$
with local coefficients computed with respect to (\ref{dder}) and (\ref{ddif}%
), which preserves the distribution (\ref{whitney}) under parallel
transports.
\end{definition}

A d--connection $\mathbf{D}$\ is metric compatible with a d--metric $\mathbf{%
g}$ if $\mathbf{D}_{\mathbf{X}}\mathbf{g}=0$ for any d--vector field $%
\mathbf{X.}$

\begin{definition}
An almost symplectic d--connection $\ _{\theta }\mathbf{D}$ on $\mathbf{V}%
^{2n}$ (equivalently, we can say that a d--connection is compatible with an
almost symplectic structure $\theta )$ is defined such that $\ _{\theta }%
\mathbf{D}$ is N--adapted, i.e., it is a d--connection, and $\ _{\theta }%
\mathbf{D}_{\mathbf{X}}\theta =0,$ for any d--vector $\mathbf{X.}$
\end{definition}

For N--anholonomic manifolds of even dimension, we have the following:

\begin{theorem}
There is a unique normal d--connection
\begin{eqnarray*}
\ \widehat{\mathbf{D}} &=&\left\{ h\widehat{D}=(\widehat{D}_{k},^{v}\widehat{%
D}_{k}=\widehat{D}_{k});v\widehat{D}=(\widehat{D}_{c},\ ^{v}\widehat{D}_{c}=%
\widehat{D}_{c})\right\} \\
&=&\{\widehat{\mathbf{\Gamma }}_{\beta \gamma }^{\alpha }=(\widehat{L}%
_{jk}^{i},\ ^{v}\widehat{L}_{n+j\ n+k}^{n+i}=\widehat{L}_{jk}^{i};\ \widehat{%
C}_{jc}^{i}=\ ^{v}\widehat{C}_{n+j\ c}^{n+i},\ ^{v}\widehat{C}_{bc}^{a}=%
\widehat{C}_{bc}^{a})\},
\end{eqnarray*}%
which is metric compatible, $\widehat{D}_{k}g_{ij}=0$ and $\widehat{D}%
_{c}g_{ij}=0,$ and completely defined by $\mathbf{g}$ and a prescribed $%
\mathcal{L}(x,y).$
\end{theorem}

\begin{proof}
Choosing
\begin{equation}
\widehat{L}_{jk}^{i}=\frac{1}{2}g^{ih}\left( \mathbf{e}_{k}g_{jh}+\mathbf{e}%
_{j}g_{hk}-\mathbf{e}_{h}g_{jk}\right) ,\widehat{C}_{jk}^{i}=\frac{1}{2}%
g^{ih}\left( \frac{\partial g_{jh}}{\partial y^{k}}+\frac{\partial g_{hk}}{%
\partial y^{j}}-\frac{\partial g_{jk}}{\partial y^{h}}\right) ,  \label{cdcc}
\end{equation}%
we construct such $\ $\ a d--connection $\widehat{\mathbf{D}}_{\alpha }=(%
\widehat{D}_{k},\widehat{D}_{c}),$ with N--adapted coefficients $\widehat{%
\mathbf{\Gamma }}_{\beta \gamma }^{\alpha }=(\widehat{L}_{jk}^{i},\ ^{v}%
\widehat{C}_{bc}^{a}).$ $\ \square $
\end{proof}

\vskip5pt

We provide the N--adapted formulas for torsion and curvature of the normal
d--connection in Appendix \ref{asect1}.

For the purposes of this work,  this property of the
normal d--connection is very important (it follows from a straightforward verification):

\begin{theorem}
\label{th3b}The normal d--connection $\widehat{\mathbf{D}}$ defines a unique
almost symplectic d--connection, $\widehat{\mathbf{D}}\equiv \ _{\theta }%
\widehat{\mathbf{D}},$ which is N--adapted, i.e. it preserves under
parallelism the splitting (\ref{whitney}), $_{\theta }\widehat{\mathbf{D}}_{%
\mathbf{X}}\theta \mathbf{=}0$ and $\widehat{T}_{jk}^{i}=\widehat{T}%
_{bc}^{a}=0,$ see (\ref{cdtors}).
\end{theorem}

We note that the normal d--connection $\widehat{\mathbf{\Gamma }}_{\beta
\gamma }^{\alpha }$ is a N--anholonomic analog of the affine connection $\
^{K}\mathbf{\Gamma }_{\beta \gamma }^{\alpha }$ and Nijenhuis tensor $^{K}%
\mathbf{\Omega }_{\ \beta \gamma }^{\alpha }$ with the torsion satisfying
the condition $\ ^{K}\mathbf{T}_{\ \beta \gamma }^{\alpha }=(1/4)^{K}\mathbf{%
\Omega }_{\ \beta \gamma }^{\alpha }$ considered in Ref. \cite{karabeg1}
(those constructions were not for spaces enabled with N--connection
structure). On N--anholonomic manifolds, we can work equivalently with both
types of linear connections.

\section{Nonholonomic Fedosov Quantization}

\label{snfq} We modify the Fedosov method in a form which allows us to
quantize the Lagrange--Finsler and Einstein spaces and related
generalizations for Einstein--Finsler quantum gravity models. Nonholonomic
Chern--Weil homomorphisms are defined. Our former results on nonholonomic
deformation quantization from \cite{vqgr2,vegpla,vlfedq} are revised with
the aim to elaborate in the next section a N--adapted trace density map
formalism (originally considered in \cite{feigin,chen}) and  prove the local
index theorem for N--anholonomic manifolds.

\subsection{Nonholonomic Chern--Weil homomorphisms}

\label{ss31}Let us consider a complex vector space $\mathcal{V}^{2n},\dim $ $%
\mathcal{V}^{2n}=2n,$ endowed with a symplectic form $b^{\alpha \beta
}=-b^{\beta \alpha };$ coefficients are defined with respect to a local \
coordinate basis $\ e_{\alpha }=\ \partial _{\alpha }=\partial /\partial
u^{\alpha },$ for $\alpha ,\beta ,...=1,2,...,2n.$ We call $\ ^{d}\mathcal{V}%
=(h\mathcal{V},v\mathcal{V})=\mathcal{V}^{2n}$ to be a distinguished vector
(d--vector) space if it is enabled with a conventional h-- and v--splitting
of type (\ref{whitney}) into two symplectic subspaces with $b^{\alpha \beta
}=(b^{ij}=-b^{ji},b^{ac}=-b^{ca})$ and $z^{\alpha }=(z^{i},z^{a}).$ On
N--anholonomic manifolds, there are considered (distinguished) d-vector,
d-tensor, d-connection, d-spinor, d-group, d-algebra etc fields which
can be adapted to N--connection structure into certain ''irreducible''
h-v--components.

Deformation quantization is elaborated for formal extensions on a ''small''
parameter $v.$
In various papers, authors prefer to write $v=\hbar $ considering certain
analogy with the Plank constant. We associate to a symplectic d--vector
space $\ ^{d}\mathcal{V}:$

\begin{definition}
The Weyl d--algebra $\ ^{d}\mathcal{W}=(h\mathcal{W},v\mathcal{W})$ over the
field $\mathbb{C}((v))$ is the d--vector space $\mathbb{C}[[\ ^{d}\mathcal{V}%
]](v)$ of \ the formally completed symmetric d--algebra of \ $\ ^{d}\mathcal{%
V}$ endowed with associative multiplication (Wick product)
{\small
\begin{eqnarray}
\ ^{1}q\circ \ ^{2}q\ (z,v) &:= &\exp \left( i\frac{v}{4}\ b^{\alpha
\beta }(\frac{\delta ^{2}}{\partial z^{\alpha }\partial \ ^{1}z^{\beta }}+ \frac{\delta ^{2}}{\partial \ ^{1}z^{\beta }\partial z^{\alpha }})\right)
\times  \notag \\
&&\ ^{1}q(z,v)\ ^{2}q(\ ^{1}z,v)\mid _{z=\ ^{1}z}  \label{wickp} \\
&=&\exp \left( i\frac{v}{4}b^{ij}\ (\frac{\delta ^{2}}{\partial
z^{i}\partial \ ^{1}z^{j}}+\frac{\delta ^{2}}{\partial \ ^{1}z^{j}\partial
z^{i}})+i\frac{v}{2}(b^{ac}\frac{\delta ^{2}}{\partial z^{a}\partial \
^{1}z^{c}})\right)  \notag \\
&& \times \ ^{1}q(z,v)\ ^{2}q(\ ^{1}z,v)\mid _{z=\ ^{1}z},  \notag
\end{eqnarray}%
}
for any $\ ^{1}q,\ ^{2}q\in $ $^{d}\mathcal{W}.$
\end{definition}

In h--v--components the product (\ref{wickp}),  we use N--elongated
derivatives $\delta /\partial z^{\alpha }$ and $\delta /\partial \
^{1}z^{\alpha }$ of type (\ref{dder}) considering that, in general, we work
with nonolonomic distributions on vector spaces, algebras and manifolds/
bundles. All constructions can be redefined for usual partial derivatives with respect to local coordinate frames.  There are natural h-- and v--filtrations, correspondingly, of $h%
\mathcal{W}$ and $v\mathcal{W}$ constructed with respect to the degrees of
monomials $2[v]+[z^{i}]$ and $2v+[z^{a}].$ These filtrations define
corresponding ''N--adapted'' $(2[v]+[z^{i}])$--adic and \ $(2[v]+[z^{a}])$%
--adic topologies parametrized in the form%
\begin{eqnarray}
... &\subset &h\mathcal{W}^{1}\subset h\mathcal{W}^{0}\subset h\mathcal{W}%
^{-1}...\subset h\mathcal{W},  \label{filtr} \\
&&\mbox{ for }h\mathcal{W}^{\ ^{h}s}=\{\sum\limits_{2k+p\geq \ ^{h}s}v^{k}\
^{k}a_{i_{i}....i_{p}}(u^{\alpha })z^{i_{1}}...z^{i_{p}}\};  \notag \\
... &\subset &v\mathcal{W}^{1}\subset v\mathcal{W}^{0}\subset v\mathcal{W}%
^{-1}...\subset v\mathcal{W},  \notag \\
&&\mbox{ for }v\mathcal{W}^{\ ^{v}s}=\{\sum\limits_{2k+p\geq \ ^{v}s}v^{k}\
^{k}a_{b_{i}....b_{p}}(u^{\alpha })z^{b_{1}}...z^{b_{p}}\},  \notag
\end{eqnarray}%
for $i_{1},i_{2},...i_{p}=1,2,...,n$ and $b_{1},b_{2},...,b_{p}=n+1,...n+n.$

For Weyl d--algebras, we have to adapt to N--connections the explicit
constructions with the $(n+n)$-th Hochschild cocycle $C^{n+n}(\ ^{d}\mathcal{W%
},\ ^{d}\mathcal{W}^{\ast })$ with coefficients in the dual d--module of $%
^{d}\mathcal{W}^{\ast }=(h\mathcal{W}^{\ast },v\mathcal{W}^{\ast }).$ Let us
consider a $(n+n)$-th simplex (with standard orientation) $\bigtriangleup
_{2n}=\{(p_{1},...,p_{2n})\in \mathbb{R}^{2n};p_{1}\leq ...\leq p_{2n}\leq
1\}$ and denote the natural projection from $\ ^{d}\mathcal{W}^{\otimes
(n+n+1)}$ onto $\mathbb{C}((v))$ by $\mu _{2n}(q_{0}\otimes ...\otimes
q_{2n})=q_{0}(0)...q_{2n}(0).$ Using any elements $q_{0},q_{\alpha }\in \
^{d}\mathcal{W}$ and $\varphi \in \ ^{d}\mathcal{W}^{\otimes (n+n+1)},$ for $%
p_{0}=0,$ we provide a N--adapted generalization of the formula for $2n$--th
Hochschild cocycle \cite{feigin},%
\begin{eqnarray}
&&C^{n+n}(\ ^{d}\mathcal{W},\ ^{d}\mathcal{W}^{\ast })\ni \tau
_{n+n}(\varphi )(q_{0})  \label{ffshf} \\
&=&\mu _{2n}\left( \int\limits_{\bigtriangleup _{2n}}\prod\limits_{0\leq
\beta \leq \gamma \leq 2n}e^{v(p_{\beta }-p_{\gamma }+1/2)b_{\beta \gamma
}}\pi _{n+n}(q_{0}\otimes \varphi )\delta p_{1}\wedge ...\wedge \delta
p_{2n}\right) ,  \notag
\end{eqnarray}%
where $\pi _{n+n}(\varphi )(q_{0}\otimes ...\otimes q_{2n})=\varepsilon
^{\alpha _{1}...\alpha _{2n}}q_{0}\otimes \frac{\delta q_{1}}{\partial
z^{\alpha _{1}}}\otimes ...\otimes \frac{\delta q_{2n}}{\partial z^{\alpha
_{2n}}}$ acts on respective tensor products of nonholonomic bases of type (%
\ref{dder}) and $\varepsilon ^{\alpha _{1}...\alpha _{2n}}$ is absolutely
antisymmetric. In the formula (\ref{ffshf}), the action of antisymmetric
form $b$ on the $\alpha $--th and $\beta $--th components of $\ ^{d}\mathcal{%
W}^{\otimes (n+n+1)}$ and $q_{\gamma }\in \ ^{d}\mathcal{W}$ is denoted by
\begin{equation*}
b_{\alpha \beta }(q_{0}\otimes ...\otimes q_{2n})=b^{\alpha _{\gamma }\alpha
_{\nu }}(q_{0}\otimes ...\otimes \frac{\delta q_{\gamma }}{\partial
z^{\alpha _{\gamma }}}\otimes ...\otimes \frac{\delta q_{\nu }}{\partial
z^{\alpha _{\nu }}}\otimes ...\otimes q_{2n}).
\end{equation*}

Let us consider an associative algebra $\mathcal{A}$ with  unit over a
field $K$ with caracteristic zero and denote, for instance, by $\mathfrak{gl}%
_{k}(\ ^{d}\mathcal{W})$ the Lie algebra of $k\times k$ matrices with values
in quadratic monomials in $\ ^{d}\mathcal{W}$. Also, we  shall use matrix
algebras of type $\mathfrak{gl}_{k}(\mathcal{A})$ and dual ones, with values
in the dual module $\mathcal{A}$ $^{\ast }.$ Denoting by $C^{\bullet }(%
\mathcal{A},\mathcal{A}^{\ast })$ the Hochschild cochain complex with
coefficients in the dual module $\mathcal{A}$ $^{\ast }$ (for $\mathcal{A}%
\rightarrow \mathfrak{gl}_{k}(\mathcal{A}),$ the similar complex is written $%
C^{\bullet }(\mathfrak{gl}_{k}(\mathcal{A}),\mathfrak{gl}_{k}(\mathcal{A}%
)^{\ast })$), we can construct a chain map%
\begin{equation*}
\ ^{k}\phi :C^{\bullet }(\mathcal{A},\mathcal{A}^{\ast })\rightarrow
C^{\bullet }(\mathfrak{gl}_{k}(\mathcal{A}),\mathfrak{gl}_{k}(\mathcal{A}%
)^{\ast })
\end{equation*}%
following the formula%
\begin{eqnarray}
&&\ ^{k}\phi (\psi )(A_{1}\otimes a_{1},...,A_{r}\otimes a_{r})(A_{0}\otimes
a_{0})  \label{map2} \\
&=&\frac{1}{r!}\sum\limits_{s\in S_{r}}(-1)^{s}\psi (a_{s(1)}\otimes
...\otimes a_{s(r)})(a_{0})\ tr(A_{0}A_{s(1)}...A_{s(r)}),  \notag
\end{eqnarray}%
where $S_{r}$ denotes the group of permutations of $r$ elements, matrices $%
A_{s}\in \mathfrak{gl}_{k}(K),a_{s}\in \mathcal{A}$, and $\psi \in C^{r}(%
\mathcal{A},\mathcal{A}^{\ast }).$ The action of map on cocycle (\ref{ffshf}%
) result in the $2n$--th cocycle in the chain complex $C^{\bullet }(%
\mathfrak{gl}_{k}(\ ^{d}\mathcal{W}),\mathfrak{gl}_{k}(\ ^{d}\mathcal{W}%
)^{\ast }),$%
\begin{equation}
\ ^{k}\Theta _{n+n}=\ ^{k}\phi (\tau _{n+n}):\wedge ^{2n}\left( \mathfrak{gl}%
_{k}(\ ^{d}\mathcal{W})\right) \otimes \mathfrak{gl}_{k}(\ ^{d}\mathcal{W}%
)\rightarrow \mathbb{C(}(v)).  \label{coc2}
\end{equation}%
The cocycle (\ref{coc2}) presents a N--adapted generalization of the results
from \cite{feigin,chen} for symplectic manifolds. In our case, we construct
a trace density map for a quantum d--algebra of functions on almost K\"{a}%
hler manifolds induced by nonholonomic distributions on Lagrange and
Einstein--Finsler spaces.

The Chern--Weil homomorphisms are defined for projections of Lie algebras to
their subalgebras. N--adapted geometric constructions on nonholonomic
manifolds are with associated Lie groups/algebras and related geometric maps are
distinguished by the --connection structure. We call a Lie distinguished
algebra (d--algebra) $\ ^{d}\mathfrak{g=}$ $\ ^{h}\mathfrak{g}\oplus $ $\
^{v}\mathfrak{g}=(\ ^{h}\mathfrak{g},\ ^{v}\mathfrak{g})$ a couple of
conventional horizontal and vertical Lie algebras associated to a
N--connection splitting $T\mathbf{V}=h\mathbf{V}\oplus v\mathbf{V}$ (see
also (\ref{whitney})). A Lie d--algebra $\ ^{d}\mathfrak{g}$ may have a
d--subalgebra $\ ^{d}\mathfrak{\rho }\subset $ $\ ^{d}\mathfrak{g},$ when $\
^{d}\mathfrak{\rho }=(\ ^{h}\mathfrak{\rho }\subset \ ^{h}\mathfrak{g},\ ^{v}%
\mathfrak{\rho }\subset \ ^{v}\mathfrak{g}).$ Let us suppose there is an $\
^{h}\mathfrak{\rho }$--equivariant N--adapted projection $\ ^{N}pr:(\ ^{h}%
\mathfrak{g}\rightarrow \ ^{h}\mathfrak{\rho },\ ^{v}\mathfrak{g}\rightarrow
\ ^{v}\mathfrak{\rho })$ satisfying the properties that the $h$-- and $v$%
--components of maps commute respectively with adjoint actions of $\ ^{h}%
\mathfrak{\rho }$ and $\ ^{v}\mathfrak{\rho }$ and $\ ^{N}pr\mid _{\ ^{d}%
\mathfrak{\rho }}=Id_{\ ^{d}\mathfrak{\rho }}.$ We can characterize \ $\
^{N}pr$ by its curvature
\begin{equation}
\ ^{d}C(\ ^{1}\zeta ,\ ^{2}\zeta ):=\left[ pr(\ ^{1}\zeta ),pr(\ ^{2}\zeta )%
\right] -pr\left( [\ ^{1}\zeta ,\ ^{2}\zeta ]\right) \in Hom\left( \wedge
^{2}\ ^{d}\mathfrak{g,}\ ^{d}\mathfrak{\rho }\right) ,  \label{proj1}
\end{equation}%
for $\ ^{1}\zeta =(\ _{h}^{1}\zeta ,\ _{v}^{1}\zeta )$ and $\ ^{2}\zeta =(\
_{h}^{2}\zeta ,\ _{v}^{2}\zeta )$ which in N--adapted form splits into $h$%
--, $v$--components of curvature,%
\begin{eqnarray*}
\ ^{h}C(\ _{h}^{1}\zeta ,\ _{h}^{2}\zeta ):= &&\left[ pr(\ _{h}^{1}\zeta
),pr(\ _{h}^{2}\zeta )\right] -pr\left( [\ _{h}^{1}\zeta ,\ _{h}^{2}\zeta
]\right) \in Hom\left( \wedge ^{2}\ ^{h}\mathfrak{g,}\ ^{h}\mathfrak{\rho }%
\right) , \\
\ ^{v}C(\ _{v}^{1}\zeta ,\ _{v}^{2}\zeta ):= &&\left[ pr(\ _{v}^{1}\zeta
),pr(\ _{v}^{2}\zeta )\right] -pr\left( [\ _{v}^{1}\zeta ,\ _{v}^{2}\zeta
]\right) \in Hom\left( \wedge ^{2}\ ^{v}\mathfrak{g,}\ ^{v}\mathfrak{\rho }%
\right) .
\end{eqnarray*}

For any given adjoint invariant d--form
\begin{equation*}
\ ^{d}A=\left( \ ^{h}A\in ((S^{r}\ ^{h}\mathfrak{\rho })^{\ast })^{\ ^{h}%
\mathfrak{\rho }},\ ^{v}A\in ((S^{r}\ ^{v}\mathfrak{\rho })^{\ast })^{\ ^{v}%
\mathfrak{\rho }}\right) ,
\end{equation*}%
the formulas
\begin{eqnarray}
&&\ ^{h}\chi (\ ^{h}A)(\ _{h}^{1}\zeta _{s},\ _{h}^{2}\zeta _{s},...,\
_{h}^{2r}\zeta _{s})=  \label{cwnh} \\
&&{\qquad }\frac{1}{(2r)!}\sum\limits_{s\in S_{2r}}(-1)^{s}\ ^{h}A\left( \
^{h}C(\ _{h}^{1}\zeta _{s},\ _{h}^{2}\zeta _{s}),...,\ ^{h}C(\
_{h}^{2r-1}\zeta _{s},\ _{h}^{2r}\zeta _{s})\right) ,  \notag \\
&&\ ^{v}\chi (\ ^{v}A)(\ _{v}^{1}\zeta _{s},\ _{v}^{2}\zeta _{s},...,\
_{v}^{2r}\zeta _{s})=  \notag \\
&&{\qquad }\frac{1}{(2r)!}\sum\limits_{s\in S_{2r}}(-1)^{s}\ ^{v}A\left( \
^{v}C(\ _{v}^{1}\zeta _{s},\ _{v}^{2}\zeta _{s}),...,\ ^{v}C(\
_{v}^{2r-1}\zeta _{s},\ _{v}^{2r}\zeta _{s})\right)  \notag
\end{eqnarray}%
define a relative Lie d--algebra cocycle $\ ^{d}\chi (\ ^{d}A)=\ ^{h}\chi (\
^{h}A)+\ ^{v}\chi (\ ^{v}A)\in C^{2r}(\ ^{d}\mathfrak{g,}\ ^{v}\mathfrak{%
\rho }).$ We say that such maps determine a Chern--Weil N--adapted
homomorphism inducing also a map from the d--vector space $\left( (S^{r}\
^{d}\mathfrak{\rho })^{\ast }\right) ^{\ ^{d}\mathfrak{\rho }}$ to $H^{2r}(\
^{d}\mathfrak{g,}\ ^{d}\mathfrak{\rho }).$ Such constructions do not depend
on the type of $\ ^{N}pr$ and/or N--connection splitting.

Taking $\ ^{d}\mathfrak{g=gl}_{k}(\ ^{d}\mathcal{W})=\ \mathfrak{gl}_{k}(\ h%
\mathcal{W})\oplus \ \mathfrak{gl}_{k}(\ v\mathcal{W}),\ ^{d}\mathfrak{\rho
=gl}_{k}\oplus \mathfrak{sp}_{n+n},$ when the Lie d--algebra $\mathfrak{sp}%
_{n+n}$ is realized as a subalgebra of scalar matrices in $\mathfrak{gl}%
_{k}(\ ^{d}\mathcal{W})$ with values in quadratic monomials in $\ ^{d}%
\mathcal{W}$, \ we can define the projection $\ ^{N}pr:\ ^{d}\mathfrak{g}%
\rightarrow \ ^{d}\mathfrak{\rho }$ following the formulas%
\begin{eqnarray}
\ ^{N}pr(\zeta ) &=&\ _{0}^{N}pr(\zeta )+\ _{2}^{N}pr(\zeta ),  \label{proj2}
\\
\ _{0}^{N}pr(\zeta ) &=&\zeta _{\mid z=0},\ \ _{2}^{N}pr(\zeta )=\frac{1}{k}%
\sigma _{2}(tr(\zeta )){I}_{k},  \notag
\end{eqnarray}%
where ${I}_{k}$ is the $k\times k$ identity matrix  and $\sigma _{2}$
is the projection onto the monomials of second degree in $z$--variables.

An explicit computation (the original constructions are due to the
Feigin--Felder--Shoikhet theorem \cite{feigin}) allows us to evaluate the
N--adapted action of (\ref{coc2}) on the identity matrix $I_{k}\in $ $\ ^{d}%
\mathfrak{g}$. This provides the proof for

\begin{theorem}
\label{th1a}There is a relative, with respect to distinguished subalgebra $\
^{d}\mathfrak{\rho =gl}_{k}\oplus \mathfrak{sp}_{n+n},$ Lie d--algebra
cocycle $\ ^{k}\Theta _{n+n}\in C^{n+n}(\ ^{d}\mathfrak{g},\ ^{d}\mathfrak{g}%
^{\ast })$ $\ $parametrized by N--adapted maps $\ ^{N}\varphi $ \ determined
by (\ref{coc2}), when
\begin{equation*}
C^{n+n}(\ ^{d}\mathfrak{g},\ \ ^{d}\mathfrak{\rho })\ni \ ^{N}\varphi =\
^{k}\Theta _{n+n}(\cdot ,\ldots ,\cdot ,I_{k}):\wedge ^{2n}(\ ^{d}\mathfrak{g%
})\rightarrow \mathbb{C}((v)).
\end{equation*}%
The cohomology class of such nonholonomic maps%
\begin{equation*}
\left[ \ ^{N}\varphi \right] =\left[ \ ^{d}\chi (\ ^{d}A_{r})=\ ^{h}\chi (\
^{h}A_{r})+\ ^{v}\chi (\ ^{v}A_{r})\right]
\end{equation*}%
coincides with the image of the $r$--th component $\ ^{d}A_{r}\in \left(
(S^{r}\ ^{d}\mathfrak{\rho })^{\ast }\right) ^{\ ^{d}\mathfrak{\rho }}$ of
the adjoint invariant d--form $\ ^{d}A\in \left( (S\ ^{d}\mathfrak{\rho }%
)^{\ast }\right) ^{\ ^{d}\mathfrak{\rho }}$, when under nonholonomic
Chern--Weil homomorphism (\ref{cwnh})%
\begin{equation*}
\ ^{d}A\left( \mathfrak{X},\ldots ,\mathfrak{X}\right) =\det \left( \frac{\
^{1}\mathfrak{X}/2v}{\sinh (\ ^{1}\mathfrak{X}/2v)}\right) ^{1/2}tr\frac{\
^{2}\mathfrak{X}}{v}
\end{equation*}%
for any d--vector $\mathfrak{X}$ with values of coefficients in the
corresponding Lie d--group, $\mathfrak{X}=$ $\ ^{1}\mathfrak{X}\oplus $ $\
^{2}\mathfrak{X}\in \mathfrak{sp}_{n+n}\oplus \mathfrak{gl}_{k}.$
\end{theorem}

For integrable distributions on spaces under considerations, the results of
this Theorem transform into ''holonomic versions'' studied in \cite%
{feigin,chen}.

\subsection{Fedosov quantization of Einstein--Finsler spaces}

Let $\ ^{\mathcal{L}}\mathbf{K}^{2n}$ be an almost K\"{a}hler space derived
for a N--anholonomic manifold $\mathbf{V}^{2n}.$ Deformation quantization of
such spaces with nontrivial torsion can be performed following methods
elaborated in Refs. \cite{karabeg1,vqgr2,vegpla,vlfedq}. In this work,
we revise the Fedosov quantization of Lagrange--Finsler and Einstein spaces
in order to include in the scheme the bi--connection formalism which can be
used for a perturbative model of quantum gravity \cite{vggr}.

We introduce the tensor $\mathbf{\Lambda }^{\alpha \beta }\doteqdot \theta
^{\alpha \beta }-i\ \mathbf{g}^{\alpha \beta },$ where $\theta ^{\alpha
\beta }$ is a form (\ref{sform}) where values with ''up'' indices are
constructed using $\ \mathbf{g}^{\alpha \beta }$ being the inverse to $%
\mathbf{g}_{\alpha \beta }$ (\ref{sasakmetr}). Considering a nonholonomic
vector bundle $\mathcal{E}$ of rank $k$ over $\ ^{\mathcal{L}}\mathbf{K}%
^{2n},$ and $\mathbf{V}^{2n},$ we denote by $End(\mathcal{E})$ \ the bundle
of endomorphisms of $\ ^{\mathcal{L}}\mathbf{K}^{2n}$ and by $End_{\mathbf{V}%
}=Sec\left( \ ^{\mathcal{L}}\mathbf{K}^{2n},End(\mathcal{E})\right) $ the
d--algebra of global sections of $End(\mathcal{E}).$

The formalism of deformation quantization can be developed by using $%
C^{\infty }(\ ^{\mathcal{L}}\mathbf{K}^{2n})[[v]]$, the space of formal
series of variable $v$ with coefficients from $C^{\infty }(\ ^{\mathcal{L}}%
\mathbf{K}^{2n})$ on a Poisson manifold $(\ ^{\mathcal{L}}\mathbf{K}%
^{2n},\{\cdot ,\cdot \}).$
In this work, we deal with the almost Poisson structure defined by the
canonical almost symplectic structure. An associative d--algebra structure
on $C^{\infty }(\ ^{\mathcal{L}}\mathbf{K}^{2n})[[v]]$ can be defined
canonically with a $v$--linear and $v$--adically continuous star product
\begin{equation}
\ ^{1}f\ast \ ^{2}f=\sum\limits_{r=0}^{\infty }\ _{r}C(\ ^{1}f,\ ^{2}f)\
v^{r},  \label{starpb}
\end{equation}%
where $\ _{r}C,r\geq 0,$ are bilinear operators on $C^{\infty }(\ ^{\mathcal{%
L}}\mathbf{K}^{2n}),$ for $\ _{0}C(\ ^{1}f,\ ^{2}f)=\ ^{1}f\ \ ^{2}f$ and $\
_{1}C(\ ^{1}f,\ ^{2}f)-\ _{1}C(\ ^{2}f,\ ^{1}f)=i\{\ ^{1}f,\ ^{2}f\};$\ $i$
being the complex unity. Such a $\ast $--operation defines a quantization
the d--algebra $End_{\mathbf{V}},$ i.e. an associative $\mathbb{C}((v))$%
--linear product in $End_{\mathbf{V}}$ $((v)),$ associated to a linear
connection $^{\ \mathcal{E}}\partial $ in $\mathcal{E},$ defined by a
1--form $\ ^{\mathcal{E}}\mathbf{\Gamma }$ with coefficients in a vector
space of dimension $k.$ Considering any bi--differential operators $\
_{r}Q:End_{\mathbf{V}}\otimes End_{\mathbf{V}}\rightarrow End_{\mathbf{V}}$
such that, for any $a,b$ $\in End_{\mathbf{V}},$ $\ _{1}Q(a,b)-$ $\
_{1}Q(b,a)=\theta ^{\alpha \beta }\ ^{\ \mathcal{E}}\partial _{\alpha }(a)\
^{\ \mathcal{E}}\partial _{\alpha }(b),$ we define a ''total'' star product
(quantization of nonholonomic $\mathcal{E}$) as
\begin{equation}
a\ast b=ab+\sum\limits_{r=1}^{\infty }\ _{r}Q(\ a,b)\ v^{r}.  \label{starp}
\end{equation}

The product (\ref{starp}) is used for constructing a formal Wick product
\begin{equation}
a\circ b\ (z)\doteqdot \exp \left( i\frac{v}{2}\ \mathbf{\Lambda }^{\alpha
\beta }\frac{\partial ^{2}}{\partial z^{\alpha }\partial \ ^{1}z^{\beta }}%
\right) a(z)b(\ ^{1}z)\mid _{z=\ ^{1}z},  \label{fpr}
\end{equation}%
for two elements $a$ and $b$ defined by series of type
\begin{equation}
a(v,u,z)=\sum\limits_{r\geq 0,|\{\alpha \}|\geq 0}\ a_{r,\{\alpha
\}}(u)z^{\{\alpha \}}\ v^{r},  \label{formser}
\end{equation}%
where by $\{\alpha \}$ we label a multi--index. \ In terms of N--elongated
derivatives, the product (\ref{fpr}) \ can be written similarly to (\ref%
{wickp}) for the Weyl d--algebra $\ ^{d}\mathcal{W}$. \ In (\ref{formser}), $%
z^{\alpha }$ are fiber coordinates of the tangent bundle $T\ \ ^{\mathcal{L}}%
\mathbf{K}^{2n}\simeq T\mathbf{V}^{2n}$ and $a_{r,\{\alpha \}}(u)=a_{r,\
^{1}\alpha \ ^{2}\alpha \ldots \ ^{l}\alpha }(u)$ can be represented as
sections of $End(\mathcal{E})\otimes S^{l}(T^{\ast }\mathbf{V}^{2n}).$

The formulas (\ref{starp})--(\ref{formser}) define a formal Wick algebra $%
\mathbf{W}_{u}$ associated with the tangent space $T_{u}\mathbf{V}^{2n},$
for $u\in \mathbf{V}^{2n}.$ The fibre product (\ref{fpr}) can be trivially
extended to the space of $\ \mathbf{W}$--valued N--adapted differential
forms $\mathcal{W}\otimes \Lambda $ by means of the usual exterior product
of the scalar forms $\Lambda ,$ where $\ \mathcal{W}$ denotes the sheaf of
smooth sections of $\mathbf{W.}$ There is a standard grading on $\Lambda $
denoted $\deg _{a}.$ We can introduce gradings $\deg _{v},\deg _{s},\deg
_{a} $ on $\ \mathcal{W}\otimes \Lambda $ defined on homogeneous elements $%
v,z^{\alpha }$ and $\mathbf{e}^{\alpha }$ (\ref{dder}) as follows: $\deg
_{v}(v)=1,$ $\deg _{s}(z^{\alpha })=1,$ $\deg _{a}(\mathbf{e}^{\alpha })=1,$
and all other gradings of the elements $v,z^{\alpha },\mathbf{e}^{\alpha }$
are set to zero. The product $\circ $ from (\ref{fpr}) on $\ \mathcal{W}%
\otimes \mathbf{\Lambda }$ is bigraded. This is written w.r.t the grading $%
Deg=2\deg _{v}+\deg _{s}$ and the grading $\deg _{a}.$ The filtration (\ref%
{filtr}) of the Weyl d--algebra gives also a natural filtration of the
nonholonomic Weyl d--algebra bundle $\ \mathbf{W}(End_{\mathbf{V}})$ whose
sections are the formal power series (\ref{formser}). The N--connection
structure on $T\mathbf{V}^{2n}$ determines also h-- and v--splitting with
filtrations in $h\mathbf{W}(End_{\mathbf{V}})$ \ and $\ v\mathbf{W}(End_{%
\mathbf{V}}).$ For holonomic configurations the algebraic and topological
properties are examined in \cite{chen}.

In what follows we shall use the d--algebra $\ ^{d}\mathbf{\Omega }^{\bullet
}(\ ^{\mathcal{L}}\mathbf{K}^{2n})$ of exterior d--forms on $\ ^{\mathcal{L}}%
\mathbf{K}^{2n}$ as an d--algebra embedded into $\ ^{d}\mathbf{\Omega }%
^{\bullet }(\mathbf{W}(End_{\mathbf{V}})).$ For any exterior d--form $\varpi
\in \ ^{d}\mathbf{\Omega }^{\bullet }(\ ^{\mathcal{L}}\mathbf{K}^{2n}),$ the
is a map sending it into the scalar matrix $\varpi I_{k}\in $ $\ ^{d}\mathbf{%
\Omega }^{\bullet }(\mathbf{W}(End_{\mathbf{V}})).$ This natural map, $i:\
^{d}\mathbf{\Omega }^{\bullet }(\ ^{\mathcal{L}}\mathbf{K}^{2n})\rightarrow $
$\ ^{d}\mathbf{\Omega }^{\bullet }(\mathbf{W}(End_{\mathbf{V}})),$ defines
the embedding of d--forms.

We consider a d--connection $\mathbf{D}$ on $T\mathbf{V}^{2n}$ which is
compatible to the symplectic d--form $\theta ^{\alpha \beta }$ (\ref{sform}%
), $\mathbf{D}\theta =0,$ and a connection $^{\ \mathcal{E}}\partial $ on $%
\mathcal{E}.$ There is a linear d--operator $\ ^{d}\mathcal{D}:\ ^{d}\mathbf{%
\Omega }^{\bullet }(\mathbf{W}(End_{\mathbf{V}}))\rightarrow \ ^{d}\mathbf{%
\Omega }^{\bullet +1}(\mathbf{W}(End_{\mathbf{V}})),$ i.e.%
\begin{equation}
\ ^{d}\mathcal{D}=\delta u^{\alpha }\left( \mathbf{e}_{\alpha }-\mathbf{%
\Gamma }_{\beta \gamma }^{\alpha }(u)z^{\gamma }\frac{\delta }{\partial
z^{\alpha }}\right) +[\ ^{\mathcal{E}}\mathbf{\Gamma ,\cdot }],
\label{dfoper}
\end{equation}%
where $\mathbf{\Gamma }_{\beta \gamma }^{\alpha }$ are coefficients of $%
\mathbf{D}$ and $\ ^{\mathcal{E}}\mathbf{\Gamma }$ is the connection form of
$^{\ \mathcal{E}}\partial .$ The curvature and torsion of such operators,
for instance, for $\mathbf{D=}\widehat{\mathbf{D}}\equiv \ _{\theta }%
\widehat{\mathbf{D}},$ see (\ref{cdcc}) is computed in \cite%
{vqgr2,vegpla,vlfedq}. In our case, there is an additional term defined by $%
\ ^{\mathcal{E}}\mathbf{\Gamma }$ but this does not modify substantially the
properties of $\ ^{d}\mathcal{D}$ determined by canonical Lagrange--Finsler
d--connections. We shall write $\ ^{d}\widehat{\mathcal{D}}$ if such a
connection is induced by $\ _{\theta }\widehat{\mathbf{D}}.$

\begin{definition}
A Fedosov--Finsler normal d--connection $\ _{r}^{d}\mathcal{D}$ \ is a
nilpotent N--adapted derivation of the graded d--algebra,
\begin{equation}
\ _{r}^{d}\mathcal{D}=\ ^{d}\mathcal{D}+v^{-1}\left[ \ ^{\theta }r,\cdot %
\right] ,  \label{ffdc}
\end{equation}%
for $\ ^{\theta }r=z^{\alpha }\theta _{\alpha \beta }(u)\delta u^{\beta }+r,$
where $r$ is an element in the set of d--forms $\ ^{d}\mathbf{\Omega }^{1}(%
\mathbf{W}(End_{\mathbf{V}}).$
\end{definition}

Any derivations of a d--connection $\mathbf{D}$ and a connection $^{\
\mathcal{E}}\partial $ can be absorbed into a d--form $r\in $ $\ ^{d}\mathbf{%
\Omega }^{1}(\mathbf{W}(End_{\mathbf{V}}).$ Two potentials $\ ^{\mathcal{E}}%
\mathbf{\Gamma }$ and $\ ^{\mathcal{E}}\widetilde{\mathbf{\Gamma }}$ result
in equivalent Fedosov--Finsler d--connections (\ref{ffdc}) if
\begin{equation}
\ ^{\mathcal{E}}\widetilde{\mathbf{\Gamma }}=\mathbf{B}^{-1}\circ \left( \ ^{%
\mathcal{E}}\mathbf{\Gamma \circ B+}v\ ^{d}\mathcal{D}\right) ,
\label{condeq}
\end{equation}%
where $\mathbf{B}$ belongs to the  affine
subspace $I_{k}\oplus Sec(\mathbf{W}^{1}(End_{\mathbf{V}}))$ in \newline $Sec(\mathbf{%
W}(End_{\mathbf{V}}))$ consisting of the sums $\mathbf{B=I}_{k}+\mathbf{B}%
_{1},$ for an arbitrary d--vector  $\mathbf{B}_{1}\in Sec(\mathbf{W}^{1}(End_{%
\mathbf{V}}))$ distinguished by N--connection. This results in equality
\begin{equation*}
\ _{r}^{d}\widetilde{\mathcal{D}}=\ _{r}^{d}\mathcal{D}+\left[ \mathbf{B}%
^{-1}\circ \ _{r}^{d}\mathcal{D}\mathbf{B,\cdot }\right].
\end{equation*}

For a trivial connection $^{\ \mathcal{E}}\partial ,$ Fedosov--Finsler
normal d--connections are completely determined by fundamental geometric
objects on $\ ^{\mathcal{L}}\mathbf{K}^{2n}.$ The normal d--connection $%
\widehat{\mathbf{D}}\mathbf{=\{}\widehat{\mathbf{\Gamma }}_{\alpha \beta
}^{\gamma }\mathbf{\}}$ (\ref{cdcc}) can be extended to the d--operator
\begin{equation}
\widehat{\mathbf{D}}\left( a\otimes \lambda \right) \doteqdot \left( \mathbf{%
e}_{\alpha }(a)-u^{\beta }\ \widehat{\mathbf{\Gamma }}_{\alpha \beta
}^{\gamma }\ ^{z}\mathbf{e}_{\alpha }(a)\right) \otimes (\mathbf{e}^{\alpha
}\wedge \lambda )+a\otimes d\lambda ,  \label{cdcop}
\end{equation}%
on $\mathcal{W}\otimes \Lambda ,$ where $^{z}\mathbf{e}_{\alpha }$ is $%
\mathbf{e}_{\alpha }$ (\ref{dder}) redefined in $z$--variables.

\begin{definition}
The Fedosov d--operators $\delta $ and $\delta ^{-1}$ on $\mathcal{W}\otimes
\mathbf{\Lambda }$ are
\begin{equation}
\delta (a)=\ \mathbf{e}^{\alpha }\wedge \mathbf{\ }^{z}\mathbf{e}_{\alpha
}(a),\ \mbox{and\ }\delta ^{-1}(a)=\left\{
\begin{array}{c}
\frac{i}{p+q}z^{\alpha }\ \mathbf{e}_{\alpha }(a),\mbox{ if }p+q>0, \\
{\qquad 0},\mbox{ if }p=q=0,%
\end{array}%
\right.  \label{feddop}
\end{equation}%
where any $a\in \mathcal{W}\otimes \mathbf{\Lambda }$ is homogeneous w.r.t.
the grading $\deg _{s}$ and $\deg _{a}$ with $\deg _{s}(a)=p$ and $\deg
_{a}(a)=q.$
\end{definition}

The d--operators (\ref{feddop}) define a N--adapted map,
\begin{equation}
a=(\delta \ \delta ^{-1}+\delta ^{-1}\ \delta +\sigma )(a),  \label{doper1}
\end{equation}
where $a\longmapsto \sigma (a)$ is the projection on the $(\deg _{s},\deg
_{a})$--bihomogeneous part of $a$ of degree zero, $\deg _{s}(a)=\deg
_{a}(a)=0;$ $\delta $ is also a $\deg _{a}$--graded derivation of the
d--algebra $\left( \mathcal{W}\otimes \mathbf{\Lambda ,\circ }\right) .$

The Fedosov--Finsler normal d--connection $\ _{r}^{d}\mathcal{D}$ (\ref{ffdc}%
) can written using the d--operator $\delta ,$ treated as the Koszul
derivation of the d--algebra \newline $\ ^{d}\mathbf{\Omega }^{\bullet }(\mathbf{W}%
(End_{\mathbf{V}})).$ The corresponding Koszul N--adapted differential
$$\delta ^{-1}(a)=z^{\alpha }\ \iota \mathbf{e}_{\alpha
}\int\limits_{0}^{1}a(u,v,\tau z,\tau \delta u)\frac{d\tau }{\tau },$$ where
the symbol $\iota \mathbf{e}_{\alpha }$ is used for the contraction of an
exterior d--form with d--vector $\mathbf{e}_{\alpha },$ for $\delta ^{-1}$
extended to $Sec(\mathbf{W}(End_{\mathbf{V}}))$ by zero. Using the
d--operator (\ref{doper1}), we can verify that such a $\delta ^{-1}$ is
really a homotopy d--operator for $\delta .$ Really, we have that for any $%
a\in \ ^{d}\mathbf{\Omega }(\mathbf{W}(End_{\mathbf{V}}))$ we can express
\begin{equation}
a=\sigma (a)+\delta \ \delta ^{-1}a+\delta ^{-1}\ \delta a,  \label{aux11}
\end{equation}%
where $\sigma $ is a N--adapted natural projection
\begin{equation}
\sigma (a)=a_{\mid z=0,\delta u=0},  \label{sigmaproj}
\end{equation}%
when $a\in \ ^{d}\mathbf{\Omega }^{\bullet }(\mathbf{W}(End_{\mathbf{V}})),$
i.e. it is defined a map of $\ ^{d}\mathbf{\Omega }^{\bullet }(\mathbf{W}%
(End_{\mathbf{V}}))$ onto the d--algebra of endomorphisms $End_{\mathbf{V}%
}(v)$ preserving $h$-- and $v$--splitting.

\begin{theorem}
\label{thgf}Any (pseudo) Lagrange--Finsler/ Riemanian metric $\mathbf{g}$ (%
\ref{sasakmetr}) defines a flat normal Fedosov d--connec\-ti\-on
\begin{equation*}
\ \widehat{\mathcal{D}}\doteqdot -\ \delta +\widehat{\mathbf{D}}-\frac{i}{v}%
ad_{Wick}(r)
\end{equation*}%
satisfying the condition $\widehat{\mathcal{D}}^{2}=0,$ where the unique
element $r\in $ $\mathcal{W}\otimes \mathbf{\Lambda ,}$ $\deg _{a}(r)=1,$ $%
\delta ^{-1}r=0,$ solves the equation
\begin{equation*}
\ \delta r=\widehat{\mathcal{T}}\ +\widehat{\mathcal{R}}+\widehat{\mathbf{D}}%
r-\frac{i}{v}r\circ r
\end{equation*}%
and this element can be computed recursively with respect to the total
degree $Deg$ as follows:%
\begin{eqnarray*}
r^{(0)} &=&r^{(1)}=0,r^{(2)}=\delta ^{-1}\widehat{\mathcal{T}},r^{(3)}=\
\delta ^{-1}\left( \widehat{\mathcal{R}}+\widehat{\mathbf{D}}r^{(2)}-\frac{i%
}{v}r^{(2)}\circ r^{(2)}\right) , \\
r^{(k+3)} &=&\ \ \delta ^{-1}\left( \widehat{\mathbf{D}}r^{(k+2)}-\frac{i}{v}%
\sum\limits_{l=0}^{k}r^{(l+2)}\circ r^{(l+2)}\right) ,k\geq 1,
\end{eqnarray*}%
where by $a^{(k)}$ we denote the $Deg$--homogeneous component of degree $k$
of an element $a\in $ $\ \mathcal{W}\otimes \mathbf{\Lambda }.$
\end{theorem}

\begin{proof}
Similarly to constructions provided in \cite{fed1,fed2,karabeg1}, using
Fedosov's d--operators (\ref{feddop})  \cite{vqgr2,vegpla,vlfedq},
we verify the conditions of the theorem. Straightforward verification of
the property $\widehat{\mathcal{D}}^{2}=0,$ with formal series of type (\ref%
{formser}), for $r$, and can be performed in N--adapted form for $\widehat{%
\mathbf{D}},$ with torsion $\widehat{\mathcal{T}}$, (\ref{cart1}), and
curvature, $\widehat{\mathcal{R}}$ (\ref{cart2}). $\square $
\end{proof}

\vskip5pt

The conditions of this Theorem can be redefined for a nonholonomic vector
bundle $\mathcal{E}$ of rank $k$ over $\ ^{\mathcal{L}}\mathbf{K}^{2n}$. The
N--adapted constructions are similar  for $h$-- and $v$--components  provided in
section 5.3 of \cite{chen} and Theorem 2 and Remark 2 in \cite{fed2}. We
summarize four necessary statements in

\begin{corollary}
\label{corolrem}--\textbf{Remarks: }For a normal d--connection $\widehat{%
\mathbf{D}}\equiv \ _{\theta }\widehat{\mathbf{D}}$ (\ref{cdcc}) on $\ ^{%
\mathcal{L}}\mathbf{K}^{2n},\ $ a connection $^{\ \mathcal{E}}\partial $ in
a nonholonomic vector bundle $\mathcal{E}$ $\ $of rank $k$ over $\ ^{%
\mathcal{L}}\mathbf{K}^{2n}$ and $\ ^{v}\Omega $ \ is \ a series of closed
distinguished (by N--connections) two--forms in $v\Omega ^{2}(\ ^{\mathcal{L}%
}\mathbf{K}^{2n})$. We can derive from Theorem \ref{thgf} the following

\begin{enumerate}
\item There is a nilpotent N--adapted derivation\newline
$\ _{r}^{d}\widehat{\mathcal{D}}=\ ^{d}\widehat{\mathcal{D}}+v^{-1}\left[ \
r-\theta _{\alpha \beta }(u)z^{\beta }\mathbf{e}^{\alpha },\cdot \right] $ \
(determined respectively as/by (\ref{ffdc}), (\ref{sform}) and (\ref{ddif})
) when the distinguished Fedosov--Weyl curvature%
\begin{eqnarray}
\ ^{W}\widehat{\mathcal{C}} &=&v(\widehat{\mathcal{R}}+\ ^{V}\mathcal{R})+2\
^{d}\mathcal{D}\ \left( \ ^{\mathcal{E}}\mathbf{\Gamma }\right) +v^{-1}[\ ^{%
\mathcal{E}}\mathbf{\Gamma ,}\ ^{\mathcal{E}}\mathbf{\Gamma }]  \label{fwdc}
\\
&=&-\theta +\ ^{v}\Omega ,  \notag
\end{eqnarray}%
where $\ ^{V}\mathcal{R}$ \ is the d--curvature of $\ ^{\mathcal{E}}\mathbf{%
\Gamma ,}$ $\ ^{d}\mathcal{D}$ \ is given by (\ref{dfoper}) and the element $%
r\in \ ^{d}\mathbf{\Omega }^{1}(\mathbf{W}^{2}(End_{\mathbf{V}}))$ satisfies
the condition $\delta ^{-1}r=0.$

\item A star product (\ref{starp}) in $End_{\mathbf{V}}((v))$ is induced via
a d--vector space isomorphism
\begin{equation}
\chi :End_{\mathbf{V}}((v))\simeq \ ^{\widehat{\mathbf{D}}}Sec(\mathbf{W}%
(End_{\mathbf{V}}))  \label{isom}
\end{equation}%
from $End_{V}((v))$ to the d--algebra $\ ^{\widehat{\mathbf{D}}}Sec(\mathbf{W%
}(End_{\mathbf{V}}))$ of \ flat sections of $\mathbf{W}(End_{\mathbf{V}})$
with respect to the Fedosov--Finsler normal d--connection $\ _{r}^{d}%
\widehat{\mathcal{D}}$ \ (\ref{ffdc}).

\item A d--connection $\ _{r}^{d}\widehat{\mathcal{D}}$  is equivalent in
the sence (\ref{condeq}) to any $\ _{r}^{d}\mathcal{D}$  when respective
distinguished Fedosov--Weyl curvatures (\ref{fwdc}), $\ ^{W}\widehat{%
\mathcal{C}}$ \ and $\ ^{W}\mathcal{C}$, represent the same de Rham
cohomology class in $H_{DR}^{2}(\ ^{\mathcal{L}}\mathbf{K}^{2n})[[v]].$

\item For a trivial d--vector bundle of rank $k=1,$ the star product (\ref%
{starp}) and isomorphism $\chi $ (\ref{isom}) result in the star product (%
\ref{starpb}) when the cohomology class of $\ ^{W}\widehat{\mathcal{C}}$ \ (%
\ref{fwdc}) characterizes the star product in $C(\ ^{\mathcal{L}}\mathbf{K}%
^{2n})((v)).$ Here we emphasize that there are two classes of equivalence,
for the $h$-- and $v$--components.
\end{enumerate}
\end{corollary}

For simplicity, we provided N--adapted constructions only for the base
nonolonomic manifold $\ ^{\mathcal{L}}\mathbf{K}^{2n}$ of a vector bundle $%
\mathcal{E}$. In general, a N--connecti\-on structure can be considered for
the tangent space $T\mathcal{E}$ (we omit such considerations in this work).

\section{Algebraic Index Theorem for Lagrange and \newline
Einstein--Finsler Spaces}

\label{sasindth} Various versions and modifications of the algebraic index
theorem generalize the Atiyah--Singer theorem \cite{atiyah} from the case of
a cotangent bundle to arbitrary symplectic and Poisson manifolds \cite%
{bressler,calaque,chen,nest}. Generalizations were also considered for
nonholonomic manifolds and Lagrange-Finsler and Einstein gerbes \cite%
{vgonz}. In this section, we apply  N-adapted techniques to provide
a local index theorem for almost K\"{a}hler models of Lagrange-Finser and
Einstein spaces.

\subsection{N--adapted trace density maps}

We use the Feigin--Felder--Shoikhet (FFS) method \cite{feigin} in order to
construct a natural, in our case, N--adapted density map. For vector bundles
on symplectic manifolds such maps were considered for a proof of the local
algebraic index theorem \cite{chen}. In our approach, the constructions
should be conventionally dubbed for $h$-- and $v$--components of geometric
objects distinguished by N--connection structure.

\begin{definition}
A N--adapted (i.e. preserving a Whitney sum $T\ ^{\mathcal{L}}\mathbf{K}%
^{2n}=h\ ^{\mathcal{L}}\mathbf{K}^{2n}\oplus v\ ^{\mathcal{L}}\mathbf{K}%
^{2n} $ of type (\ref{whitney})) $\mathbb{C}$--linear nonholonomic map
\begin{equation*}
\ _{N}^{d}tr:\ End_{\mathbf{V}}((v))\rightarrow H^{2n^{\prime }}(\mathbf{K}%
^{2n})((v))
\end{equation*}%
vanishing on commutators $\ _{N}^{d}tr\left( a\ast b-b\ast a\right) =0,$
defined by star d--operator (\ref{starp}) for any $a,b\in \ End_{\mathbf{V}%
}((v)),$ is called a trace density d--map.
\end{definition}

In the above formula for $\ _{N}^{d}tr,$ we write $H^{2n^{\prime }}$ with a
prime index in order to avoid confusions with the dimension $n$ in $\mathbf{K%
}^{2n}.$

Using the nonholonomic cocycle $\ ^{k}\Theta _{n+n}$ (\ref{coc2}) of
d--group $\ ^{d}\mathfrak{g=gl}_{k}(\ ^{d}\mathcal{W})=\ \mathfrak{gl}_{k}(\
h\mathcal{W})\oplus \ \mathfrak{gl}_{k}(\ v\mathcal{W})$ relative to
N--adapted ''shifts'' $\ ^{d}\mathfrak{\rho =gl}_{k}\oplus \mathfrak{sp}%
_{n+n},$ we can construct a nonholonomic map $\ ^{\mathcal{D}}\digamma :\
End_{\mathbf{V}}((v))\rightarrow $ $\ ^{d}\mathbf{\Omega }^{2n^{\prime }}(%
\mathbf{K}^{2n})$ following the formula%
\begin{equation}
\ ^{\mathcal{D}}\digamma (a)=v^{-n}\ ^{k}\Theta _{n+n}\left( \ ^{\theta }r,%
\mathbf{...,}\ ^{\theta }r,\chi (a)\right) ,  \label{ndistm}
\end{equation}%
where $\ ^{\theta }r\in \ ^{d}\mathbf{\Omega }^{1}(\mathbf{W}(End_{\mathbf{V}%
}))$ is the distinguished 1--form considered in the \ formula for the
Fedosov--Finsler d--connection $\ _{r}^{d}\mathcal{D}$ (\ref{ffdc}) and $%
\chi (a)$ is the isomorphism (\ref{isom}).

Let us consider respectively two nonholonomic vector bundles $\ ^{1}\mathcal{%
E}$ and $\ ^{2}\mathcal{E}$ of rank $k$ over $\ ^{\mathcal{L}}\mathbf{K}%
^{2n} $ endowed with Fedosov--Finsler d--connections $\ _{r}^{d}\widehat{%
\mathcal{D}}_{1}$ and $\ _{r}^{d}\widehat{\mathcal{D}}_{2}$ \ on $\mathbf{W}%
(End_{\ ^{1}\mathbf{V}})$ and $\mathbf{W}(End_{\ ^{2}\mathbf{V}});$ elements
$\ ^{1}a\in \ End_{\ ^{1}\mathbf{V}}$ and a pair of N--adapted endomorphisms
$a,b\in \ End_{\mathbf{V}}((v))$ corresponding to $\mathcal{E} $ over $\ ^{%
\mathcal{L}}\mathbf{K}^{2n};$ two equivalent d--operators $\ ^{\mathcal{E}}%
\widetilde{\mathbf{\Gamma }}$ and $\ ^{\mathcal{E}}\mathbf{\Gamma }$ in the
sense of (\ref{condeq}) with $\mathbf{B}\in I_{k}\oplus Sec(\mathbf{W}%
^{1}(End_{\mathbf{V}}))$ and corresponding equivalent Fedosov--Finsler
d--operators $\ _{r}^{d}\mathcal{D}$ and $\ _{r}^{d}\widetilde{\mathcal{D}};$
a N--adapted projection $\sigma $ (\ref{sigmaproj}) and isomorphism $\chi $ (%
\ref{isom}); the $r$--th component $\ ^{d}A_{r}\in \left( (S^{r}\ ^{d}%
\mathfrak{\rho })^{\ast }\right) ^{\ ^{d}\mathfrak{\rho }}$; the identity
endomorphism $I_{k}\in End_{\mathbf{V}};$ we also consider curvature $\
^{d}C(\ ^{1}\zeta ,\ ^{2}\zeta )$ (\ref{proj1}) and projection (\ref{proj2}%
). Applying respectively the N--adapted constructions for $\ ^{k}\Theta
_{n+n}$ in section \ref{ss31} (see formulas (\ref{ffshf}) -- (\ref{coc2})
and Theorem \ref{th1a}) and dubbing for $h$--$v$--components the respective
computations from \cite{feigin} we prove

\begin{theorem}
\label{th4}The nonholonomic map $\ ^{\mathcal{D}}\digamma $ (\ref{ndistm})
is characterized by properties that $\ ^{\mathcal{D}_{1}\oplus \mathcal{D}%
_{2}}\digamma (\ ^{1}a\oplus 0)=\ ^{\mathcal{D}_{1}}\digamma (\ ^{1}a)$ and
three classes of N--adapted forms $\ ^{\mathcal{D}}\digamma (a\ast b-b\ast
a),\ ^{\widetilde{\mathcal{D}}}\digamma \left( \sigma \left( \mathbf{B}%
^{-1}\circ \chi (\ ^{1}a)\circ \mathbf{B}\right) \right) -\ ^{\mathcal{D}%
}\digamma (\ ^{1}a)$ and $\ ^{\mathcal{D}}\digamma \left( I_{k}\right)
-v^{n^{\prime }}$ $^{d}A_{n^{\prime }}\left( \ ^{d}C(\ ^{\theta }r,\ \
^{\theta }r),...,\ ^{d}C(\ ^{\theta }r,\ \ ^{\theta }r)\right) $ are of type
$\delta \ ^{d}\mathbf{\Omega }^{2n^{\prime }-1}(\mathbf{K}^{2n})((v)),$ \
for $0$ denoting the trivial endomorphysm of $\ ^{2}\mathcal{E}$ \ and $%
\delta $ being the de Rham differential computed in N--adapted form.
\end{theorem}

From this theorem, we get

\begin{corollary}
We define a N--adapted trace density map
\begin{equation}
\ _{N}^{d}tr(a)=\left[ \ ^{\mathcal{D}}\digamma \left( a\right) \right] :\
End_{\mathbf{V}}((v))\rightarrow H^{2n^{\prime }}(\mathbf{K}^{2n})((v))
\label{ffsd}
\end{equation}%
if the star--product $\ast $ in $End_{\mathbf{V}},$ see (\ref{starp}), is
determined by an isomorphism $\chi $ (\ref{isom}) and a normal d--connection
$\widehat{\mathbf{D}}\equiv \ _{\theta }\widehat{\mathbf{D}}$ (\ref{cdcc}).
\end{corollary}

Finally, we note that for holonomic structures \ and vanishing distortion
tensor (\ref{cdeftc}) the formula (\ref{ffsd}) transform into the so--called
FFS trace density map, see details in \cite{feigin} and section 2 of \cite%
{chen}.

\subsection{Main result:\ the nonholonomic algebraic index theorem}

We conclude our constructions on nonholonomic deformation quantization and
N--adapted index theorem for an almost K\"{a}hler space $\ ^{\mathcal{L}}%
\mathbf{K}^{2n}$ defined by a regular $\mathcal{L}$ on a N--anholonomic
manifold $\mathbf{V}^{2n}.$ Let $\ast $ be a star product (\ref{starpb}) in
the d--vector space $C(\ ^{\mathcal{L}}\mathbf{K}^{2n})((v))$ of smooth
functions on $\ ^{\mathcal{L}}\mathbf{K}^{2n},$ with induced (\ref{starp})
following conditions of the statement 2 in Corollary \ref{corolrem} via
 the isomorphism $\chi $ (\ref{isom}). We consider an idempotent $\ ^{p}\zeta $
in the matrix d--algebra $\ ^{d}\mathfrak{g}$ $=\mathfrak{gl}_{k}\left( \
_{c}^{+}\mathcal{A}\right) ,$ where the subalgebra
 $\ _{c}^{+}\mathcal{A=}C(\ ^{\mathcal{L}}\mathbf{K}^{2n})[[v]]\subset \ _{c}%
\mathcal{A=}C(\ ^{\mathcal{L}}\mathbf{K}^{2n})((v)),$ %
and use the same symbol $\ _{r}^{d}\mathcal{D}$ for the naturally extended
Fedosov--Finsler d--connection on the Weyl d--algebra bundle $\ \mathbf{W}%
(End(l_{k}))$ associated with the trivial bundle $l_{k}$ of rank $k.$ At the
next step, we assign to $\ ^{p}\zeta $ the top degree de Rham cohomology
class%
\begin{equation*}
cl(\ ^{p}\zeta )=[\ ^{\mathcal{D}}\digamma \left(\ ^{p}\zeta \right) ]\in
H^{2n^{\prime }}(\mathbf{K}^{2n})((v)).
\end{equation*}%
For the \ $K_{0}$--group of the d--algebra $\ _{c}^{+}\mathcal{A}$ to $%
H^{2n^{\prime }}(\mathbf{K}^{2n})((v))$ we get $cl:K_{0}(\ _{c}^{+}\mathcal{A%
})\rightarrow H^{2n^{\prime }}(\mathbf{K}^{2n})((v))$ following statements
of Theorem \ref{th4}.

Certain properties of a Fedosov quantized Lagrange--Finsler/ Einstein
geometric model are characterized by star product (\ref{starpb}), and (\ref%
{starp}), determined by the principal part of idempotents of $\mathfrak{gl}%
_{k}\left( \ _{c}^{+}\mathcal{A}\right) .$

\begin{proposition}
--\textbf{Definition.} The principal part $\ ^{0}\zeta $ of an idempotent $\
^{p}\zeta \in \mathfrak{gl}_{k}\left( \ _{c}^{+}\mathcal{A}\right) $ is
determined by the \ zeroth term $\ ^{0}\zeta =\ ^{p}\zeta _{|v=0}.$ There
are two important properties for such principal parts:

\begin{enumerate}
\item Any idempotend in the matrix d--algebra $\mathfrak{gl}_{k}C(\ ^{%
\mathcal{L}}\mathbf{K}^{2n})$ combined with the operation of taking the
principal part gives a well--defined principal symbol map%
\begin{equation}
\mathcal{K}:K_{0}(\ _{c}^{+}\mathcal{A})\rightarrow K_{0}(\ ^{\mathcal{L}}%
\mathbf{K}^{2n}).  \label{prinsind}
\end{equation}

\item If two idempotents $\ ^{p}\zeta _{1},$ $\ ^{p}\zeta _{2}\in \left(
End_{V}[[v]],\ast \right) ,$ for a vector bundle $\mathcal{E}$ of rank $k$
over $\ ^{\mathcal{L}}\mathbf{K}^{2n}$ and a star product $\ast $ (\ref%
{starp}) defined by Fedosov procedure, have equal principal parts, $\
^{p}\zeta _{1|v=0}=\ ^{p}\zeta _{2|v=0},$ then we get coincidence of
cohomology classes, $\ ^{\mathcal{D}}\digamma \left( \ ^{p}\zeta _{1}\right)
=\ ^{\mathcal{D}}\digamma \left( \ ^{p}\zeta _{2}\right) .$
\end{enumerate}
\end{proposition}

\begin{proof}
The first statement above follows from a result in \cite{bout}, that for
any idempotent $\zeta $ in the matrix algebra $\mathfrak{gl}_{k}\left( C(%
\mathcal{M}\right) )$, for a symplectic manifold $\mathcal{M}$, we can
define an idempotent $\zeta ^{\prime }$ in $\mathfrak{gl}_{k}\left( C(%
\mathcal{M}\right) [[v]],\ast ),$ where $\left( C(\mathcal{M}\right)
[[v]],\ast )$ is a subalgebra of $\left( C(\mathcal{M}\right) ((v)),\ast ),$
with a given star product $\ast ,$ when the principal part of $\zeta
^{\prime }$ is $\zeta .$ In our approach, we work with N--adapted structures
on $C(\ ^{\mathcal{L}}\mathbf{K}^{2n})$ $\ $and have to consider for $h$--
and $v$--components and almost symplectic d--connection $\widehat{\mathbf{D}}%
\equiv \ _{\theta }\widehat{\mathbf{D}}$ (\ref{cdcc}) on $\ ^{\mathcal{L}}%
\mathbf{K}^{2n}$ the explicit formula for principal parts in \cite{fed2}.\ \
The second statement is a straightforward generalization for $\ ^{\mathcal{L}%
}\mathbf{K}^{2n}$ of the proof of theorem 6.1.3 in the just mentioned
Fedosov's monograph. $\square $
\end{proof}

\begin{theorem}
\label{th5}The cohomology class $cl(\ ^{0}k)$ for any element $\ ^{0}k\in
K_{0}(\ _{c}^{+}\mathcal{A})$ coincides with the top component of the cup product%
\begin{equation}
cl(\ ^{0}k)=\left[ \widehat{A}(\ ^{\mathcal{L}}\mathbf{K}^{2n})\ e^{-v^{-1}\
^{W}\widehat{\mathcal{C}}}\ ch\left( \mathcal{K}(\ ^{0}k)\right) \right]
_{n+n},  \label{mainf}
\end{equation}%
where $ch\left( \mathcal{K}(\ ^{0}k)\right) $ is the the Chern character of
the principal symbol $\mathcal{K}(\ ^{0}k)$ (\ref{prinsind}) of $\ ^{0}k,$
the exponent is determined by $\ ^{W}\widehat{\mathcal{C}}=-\theta +\
^{v}\Omega $ (\ref{fwdc}) (in the holonomic case it transforms into the
so--called Deligne--Fedosov class) here computed by the star--product in $\
_{c}^{+}\mathcal{A};$ we also consider in (\ref{mainf}) the so--called $%
\widehat{A}$--genus of $\ ^{\mathcal{L}}\mathbf{K}^{2n}.$
\end{theorem}

We sketch a proof of this theorem in Appendix \ref{asproof}.

\subsection{Example: index encoding of Einstein--Finsler spaces}

It is possible to find explicit relations between the nonholonomic versions
of the local index theorem and classification of exact solutions for
Einstein and Einstein--Finsler classical and quantum (in the sense of
deformation quantization) gravity. For instance, a very important
mathematical and physical problem is that to formulate some well established
criteria when a nonholonomical distinguished 1--form $\theta $ (\ref{sform})
and related normal d--connection $\widehat{\mathbf{D}}\equiv \ _{\theta }%
\widehat{\mathbf{D}}$ (\ref{cdcc}), via the Fedosov--Finsler normal
d--connection $\ _{r}^{d}\mathcal{D}$ (\ref{ffdc}), in the product (\ref%
{mainf}), define cohomology classes for solutions of gravitational field
equations in various types of gravity theories. For simplicity, we consider
a four dimensional (4-d) nonholonomic manifold $\mathbf{V}$ with local
coordinates $u^{\alpha }=(x^{1},x^{2},y^{3}=t,y^{4})$ and signature $(++-+),$
i.e. $t$ is a time like coordinate. The Einstein equations with
N--anholonomic distributions on pseudo--Riemannian and/or Lagrange--Finsler
spaces with canonical d--connection/ normal $h$--$v$--connection / Cartan
d--connection can be integrated in very general forms, see details and
proofs in Refs. \cite{vgensol,vrflg}. In this section, we show
how general classes of exact solutions of Einstein equations can be generated
as almost K\"{a}hler structures and classified following conditions of the
algebraic index Theorem \ref{th5}.

\subsubsection{Einstein--Finsler and Einstein spaces}

Einstein spaces, with the Ricci tensor proportional to the metric tensor via
a nontrivial cosmological constant, and its polarizations can be constructed
in various types of gravity theories. They are used for different studies of
properties of gravitational vacuum and/or as the simplest approximations for
matter, or extra dimension, contributions. Using N--connection distributions
and nonholonomic transforms, we can model Finsler configurations as exact
solutions of gravitational field equations in general relativity and,
inversely, a Finsler gravity model on tangent bundle can be constructed
similarly to the Einstein gravity but for a different class of Finsler
d--connections, see details and discussions in \cite{vcrit,vrflg}%
.

\begin{proposition}
--\textbf{Definition.}

\begin{enumerate}
\item A {nonholonomic Einstein space} for the normal d--connection $\widehat{%
\mathbf{D}}$ (\ref{cdcc}) and d--metric $\mathbf{g}_{\alpha \beta }$ (\ref%
{sasakmetr}) is defined by solutions of Einstein equations (\ref{ensteqcdc})
with source%
\begin{equation}
\mathbf{\Upsilon }_{\ \delta }^{\alpha }=diag[\ ^{v}\lambda (x^{k},t),\
^{v}\lambda (x^{k},t);\ ^{h}\lambda (x^{k}),\ ^{h}\lambda (x^{k})],
\label{source}
\end{equation}%
\begin{equation}
\widehat{\mathbf{R}}_{\ \beta \delta }=\mathbf{\Upsilon }_{\ \delta
}^{\alpha }.  \label{efspaces}
\end{equation}

\item We construct geometric and physicial models of {Einstein--Finsler/--Lagrange spaces} if $\ \widehat{%
\mathbf{D}}$ and $\mathbf{g}_{\alpha \beta }=\ ^{F}\mathbf{g}_{\alpha \beta
},$ or $\ =\ ^{L}\mathbf{g}_{\alpha \beta },$ are taken for a
Finsler/--Lagrange geometry. Almost K\"{a}hler models of
Einstein/--Finsler/--Lagrange spaces are elaborated in variables $\mathbf{%
\theta (\cdot ,\cdot )}\doteqdot \mathbf{g}\left( \mathbf{J\cdot ,\cdot }%
\right) $ and $\widehat{\mathbf{D}}\equiv \ _{\theta }\widehat{\mathbf{D}}$
for correspondingly prescribed N--connection structures.

\item As a particular case, we extract {Einstein spaces} for the
Levi--Civita connection $\nabla $ if we impose additionally the conditions (%
\ref{lccond}) and consider sources $\ ^{v}\lambda =\ ^{h}\lambda =\lambda
=const.$
\end{enumerate}
\end{proposition}

\begin{proof}
The first statement with equations (\ref{efspaces}) follows by contraction
of indices in equations (\ref{ensteqcdc}) with source (\ref{source}). Such
equations define ''standard'' Einstein--Finsler spaces if $\widehat{\mathbf{D%
}}$ and $\mathbf{g}_{\alpha \beta }$ are considered on $TM.$

The second statement can be derived for any frame/coordinate transforms
\begin{equation}
\mathbf{f}_{\alpha \beta }=\ \mathbf{e}_{\ \alpha }^{\alpha ^{\prime }}\
\mathbf{e}_{\ \beta }^{\beta ^{\prime }}\mathbf{g}_{\alpha ^{\prime }\beta
^{\prime }},  \label{frametr}
\end{equation}%
where $\mathbf{f}_{\alpha \beta }$ is a Finsler/ Lagrange metric of type (%
\ref{sasakmetr}), induced by a Hessian (\ref{lm}), and $\mathbf{g}_{\alpha
^{\prime }\beta ^{\prime }}$ is a solution of (\ref{efspaces}). In 4--d
spaces, we have to find 16 coefficients $\mathbf{e}_{\ \alpha }^{\alpha
^{\prime }},$ with given maximum 6 independent coefficients $\mathbf{g}%
_{\alpha ^{\prime }\beta ^{\prime }}.$ By local coordinate transforms, and
because of Bianchi identities, we can put zero 4 coefficients from 10 ones
of a second rank symmetric tensor. In such way, we can also such way chose a generating
Lagrange/Finsler function $L(x,y),$ or $\mathcal{L}(x,y),$ when the
equations (\ref{frametr}) can be solved on local cartes which should be neighborhoods, or charts throughout.
 Values $\mathbf{e}_{\ \alpha }^{\alpha ^{\prime }}$ together
with N--connection coefficients $N_{i}^{a}$ induced by a corresponding $L/$ $%
\mathcal{L}$ state the nonholonomic distribution which is considered for our
Einstein--Finsler spacetime model. We conclude that a solution $\mathbf{g}$
of Einstein equations (\ref{efspaces}) can be written in N--adapted
nonholonomic variables for some data $\left(g_{ij},h_{ab},N_{i}^{a}\right) $
encoding $\mathbf{g}_{\alpha \beta }$ in local coordinates, or
(equivalently) in Finsler variables $\left(\mathbf{f}_{\alpha \beta },\
^{c}N_{i}^{a},\widehat{\mathbf{D}}\right) $, via frame transform (\ref%
{fralgeq}) and (\ref{frametr}). The priority of Finsler variables in various
models of gravity is that choosing $\mathbf{\theta (\cdot ,\cdot )}\doteqdot
\mathbf{g}\left( \mathbf{J\cdot ,\cdot }\right) $ and $\widehat{\mathbf{D}}%
\equiv \ _{\theta }\widehat{\mathbf{D}}$ we uniquely define  compatible
with $\mathbf{\theta ,}\ _{\theta }\widehat{\mathbf{D}}\mathbf{\theta} =0,$
almost K\"{a}hler geometric models. This allows us to apply the Fedosov
quantization methods.

The third statement is a consequence of the fact that the constraints $%
\widehat{L}_{aj}^{c}=e_{a}(N_{j}^{c}),\ \widehat{C}_{jb}^{i}=0,\ \Omega _{\
ji}^{a}=0$ (\ref{lccond}) are satisfied for the tensors $\widehat{\mathbf{T}}_{\
\alpha \beta }^{\gamma }$(\ref{cdtors}) and then these $Z_{\ \alpha \beta }^{\gamma }$(%
\ref{cdeft}) are zero. This states that $\widehat{\mathbf{\Gamma }}_{\
\alpha \beta }^{\gamma }=\Gamma _{\ \alpha \beta }^{\gamma },$ with respect
to N--adapted frames (\ref{dder}) and (\ref{ddif}), see (\ref{cdeftc}), even
$\widehat{\mathbf{D}}\neq \nabla .$ $\square $
\end{proof}

\subsubsection{On ''general'' exact solutions in gravity}

Let us consider a metric parameterized with respect to a
N--adapted cobase (\ref{ddif}) in the form
\begin{eqnarray}
\ ^{\circ }\mathbf{g} &\mathbf{=}&e^{\psi (x^{k})}{dx^{i}\otimes dx^{i}}%
+h_{3}(x^{k},t)\mathbf{e}^{3}{\otimes }\mathbf{e}^{3}+\omega
^{2}(x^{k},t,x^{4})h_{4}(x^{k},t)\mathbf{e}^{4}{\otimes }\mathbf{e}^{4},
\notag \\
\mathbf{e}^{3} &=&dt+w_{i}(x^{k},t)dx^{i},\mathbf{e}%
^{4}=dy^{4}+n_{i}(x^{k},t)dx^{i}.  \label{genans}
\end{eqnarray}%
In brief, we denote $\partial a/\partial
x^{1}=a^{\bullet },\partial a/\partial x^{2}=a^{\prime }$ and $\partial
a/\partial t=a^{\ast }.$
\begin{proposition}
For $h_{3,4}^{\ast }\neq 0,$ the system of gravitational field equations
defining Einstein--Finsler spaces for a normal d--connection and a metric (\ref{genans}) transform into
\begin{eqnarray}
\ddot{\psi}+\psi ^{\prime \prime } &=&2\ ^{h}\lambda (x^{k}),  \label{4ep1a}
\\
h_{4}^{\ast } &=&2h_{3}h_{4}\ ^{v}\lambda (x^{i},t)/\phi ^{\ast },
\label{4ep2a} \\
\beta w_{i}+\alpha _{i} &=&0,  \label{4ep3a} \\
n_{i}^{\ast \ast }+\gamma n_{i}^{\ast } &=&0,  \label{4ep4a} \\
\mathbf{e}_{k}\omega =\partial _{4}\omega +w_{k}\omega ^{\ast
}+n_{k}\partial \omega /\partial y^{4} &=&0  \label{4ep5a}
\end{eqnarray}%
where%
\begin{equation}
~\phi =\ln |\frac{h_{4}^{\ast }}{\sqrt{|h_{3}h_{4}|}}|,\ \alpha
_{i}=h_{4}^{\ast }\partial _{i}\phi ,\ \beta =h_{4}^{\ast }\ \phi ^{\ast },\
\gamma =\left( \ln |h_{4}|^{3/2}/|h_{3}|\right) ^{\ast }.  \label{auxphi}
\end{equation}
\end{proposition}

\begin{proof}
For $\omega =1,$ general proofs are contained in Refs. \cite%
{vrflg}. In \cite{vgensol}, there are provided\
respectively general solutions for nontrivial $\omega (x^{k},t,x^{4})$ and
for the so--called Cartan connection in Finsler geometry. Such a proof
consists of straightforward computations of the Ricci d--tensors for the
d--connections under consideration, and equation (\ref{genans}). For some
special cases when $h_{3}^{\ast }=0,$ or $h_{4}^{\ast }=0,$ and/or $\phi
^{\ast }=0,$ similar systems of equations can be derived. The rest of the proof is a lengthy and cumbersome computation that we leave to the reader. $\square $
\end{proof}

\vskip5pt

The above system of equations is with splitting of equations (not be
confused with splitting of variables) which allows us to construct exact
solutions in very general forms.

\begin{theorem}
\label{th6}If a metric $\mathbf{g}_{\alpha \beta }$ (\ref{sasakmetr}) in
general relativity and/or Einstein--Finsler gravity can be related via
nonholonomic transform to an ansatz (\ref{genans}), such a metric defines
respectively an Einstein and/or Einstein--Finsler space.
\end{theorem}

\begin{proof}
We sketch the proof \ of 4--d and conditions $h_{3,4}^{\ast }\neq 0$ (in %
\cite{vgensol}, there are provided formulas for arbitrary dimensions
for different classes of d--connections). If $h_{4}^{\ast }\neq 0;\Upsilon
_{2}\neq 0,$ we get $\phi ^{\ast }\neq 0.$ Prescribing any nonconstant $\phi
=\phi (x^{i},t)$ as a generating function, we can construct exact solutions
of (\ref{4ep1a})--(\ref{4ep4a}): We solve step by step the two dimensional
Laplace equation, for $g_{1}=g_{2}=e^{\psi (x^{k})};$ integrate on $t,$ in
order to define $h_{3},$ $h_{4}$ and $n_{i};$ and finally solve the
algebraic equations, for $w_{i}.$ Finally,  the solutions are obtained
(computing consequently for a chosen $\phi (x^{k},t)$)
\begin{eqnarray}
g_{1} &=&g_{2}=e^{\psi (x^{k})},h_{3}=\pm \ \frac{|\phi ^{\ast }(x^{i},t)|}{%
\ ^{v}\lambda (x^{i},t)},\   \label{gsol1} \\
h_{4} &=&\ ^{0}h_{4}(x^{k})\pm \ 2\int \frac{(\exp [2\ \phi
(x^{k},t)])^{\ast }}{\ ^{v}\lambda (x^{i},t)}dt,\   \notag \\
w_{i} &=&-\partial _{i}\phi (x^{i},t)/\phi ^{\ast }(x^{i},t),\   \notag \\
n_{i} &=&\ ^{1}n_{k}\left( x^{i}\right) +\ ^{2}n_{k}\left( x^{i}\right) \int
[h_{3}(x^{i},t)/(\sqrt{|h_{4}(x^{i},t)|})^{3}]dt,  \notag
\end{eqnarray}%
where $\ ^{0}h_{4}(x^{k}),\ ^{1}n_{k}\left( x^{i}\right) $ and $\
^{2}n_{k}\left( x^{i}\right) $ are integration functions. In these formulas,
we have to fix a corresponding sign $\pm $ in order to generate a necessary
local signature of type $(++-+)$ for some chosen $\phi ,\Upsilon _{2}$ and $%
\Upsilon _{4}.$ The function $\omega ^{2}(x^{k},t,x^{4})$ can be an
arbitrary one constrained to the condition (\ref{4ep5a}). Such d--metrics
generate Einstein--Finsler spaces for metric compatible Finsler
d--connections.

To extract exact solutions in general relativity, i.e. for the Levi--Civita
connection, we have to constrain the coefficients (\ref{gsol1}) of metric (%
\ref{genans}) to satisfy the conditions (\ref{lccond}) and consider sources $%
\ ^{v}\lambda =\ ^{h}\lambda =\lambda =const.$ This imposes additional
constraints on the classes of generating and integration functions. We can
select a subclass of Einstein spaces when $\ ^{2}n_{k}\left( x^{i}\right) =0$
and $\ ^{1}n_{k}\left( x^{i}\right) $ are subjected to conditions $\
\partial _{i}\ ^{1}n_{k}=\partial _{k}\ ^{1}n_{i}. $ For $w_{i}=-\partial
_{i}\phi /\phi ^{\ast },$ we get \ functional constraints on $\phi
(x^{k},t), $ when
\begin{eqnarray}
\left( w_{i}[\phi ]\right) ^{\ast }+w_{i}[\phi ]\left( h_{4}[\phi ]\right)
^{\ast }+\partial _{i}h_{4}[\phi ]=0, &&  \notag \\
\partial _{i}\ w_{k}[\phi ]=\partial _{k}\ w_{i}[\phi ], &&  \label{auxc1}
\end{eqnarray}%
where, for instance, we denoted by $h_{4}[\phi ]$ the functional dependence
on $\phi .$

Finally, we emphasize that the generic off--diagonal ansatz (\ref{genans})
define a very general class of exact solutions of gravitational field
equations depending on all coordinates. Any metric related by frame
transforms with a solution (\ref{gsol1}), $\ g_{\alpha \beta }=\ \mathbf{e}%
_{\ \alpha }^{\alpha ^{\prime }}\ \mathbf{e}_{\ \beta }^{\beta ^{\prime }}\
^{\circ }\mathbf{g}_{\alpha ^{\prime }\beta ^{\prime }}$ also defines an
exact solution. And inversely, for very general assumptions, if a metric $%
g_{\alpha \beta }$ is a solution of Einstein/--Finsler equations, such a
metric can be such a way parametrized by certain prescribed N--anholonomic
distributions that an ansatz  $\ ^{\circ }\mathbf{g}_{\alpha ^{\prime
}\beta ^{\prime }}$ will be constructed (with corresponding generating and
integration functions). $\square $
\end{proof}

\subsubsection{Almost K\"{a}hler--Finsler variables}

In this subsection we show that for any given solution of Einstein equations,
$\ ^{\circ }\mathbf{g}$ (\ref{genans}), on a nonholonomic manifold/bundle $%
\mathbf{V,}$ $\dim \mathbf{V}=4,$ we can introduce such a parametrization
for the nonholonomic structure when the corresponding N--adapted geometric
objects induce an almost K\"{a}hler structure.
We put a left label ''$\ ^{\circ }"$ to values determining an exact solution
for an Einstein and/or Einstein--Finsler space.

Finsler variables \ are introduced as solutions (on any chart or neighborhood) of algebraic equations $\mathbf{f}_{\alpha \beta }=\
\mathbf{e}_{\ \alpha }^{\alpha ^{\prime }}\ \mathbf{e}_{\ \beta }^{\beta
^{\prime }}\ ^{\circ }\mathbf{g}_{\alpha ^{\prime }\beta ^{\prime }}$ (\ref%
{frametr}), where $\mathbf{f}_{\alpha \beta }$ is a Sasaki type d--metric (%
\ref{sasakmetr}) with coefficients generated by a regular $L=F^{2},$ or $%
\mathcal{L=F}^{2},$ in the form $f_{ab}$ (\ref{lm}) and $\ ^{c}N_{i}^{a}$ (%
\ref{clnc}) and $\ ^{\circ }\mathbf{g}_{\alpha ^{\prime }\beta ^{\prime }}$
is a general solution of type (\ref{genans}). We can fix the
parameterizations, fix certain types of generating and integration functions,
additional frame/coordinate transforms etc when some solutions of $\ \mathbf{%
e}_{\ \alpha }^{\alpha ^{\prime }}$ are in a ''simple'' diagonal form. For
instance, we can write in explicit h-- and v--components
\begin{eqnarray}
\ f_{ij} &=&e_{\ i}^{i^{\prime }}e_{\ j}^{j^{\prime }}\ ^{\circ
}g_{i^{\prime }j^{\prime }}\mbox{\ and  \ }~\ f_{ab}=e_{\ a}^{a^{\prime
}}e_{\ b}^{b^{\prime }}\ \ ^{\circ }g_{a^{\prime }b^{\prime }},
\label{auxeq1} \\
\ \ ^{\circ }N_{i^{\prime }}^{a^{\prime }} &=&e_{i^{\prime }}^{\ i}e_{\
a}^{a^{\prime }}\ ^{c}N_{i}^{a},\mbox{\ or \ }\ ^{c}N_{i}^{a}=e_{i}^{\
i^{\prime }}e_{\ a^{\prime }}^{a}\ \ ^{\circ }N_{i^{\prime }}^{a^{\prime }},
\label{auxeq2}
\end{eqnarray}%
were, for instance, $e_{a^{\prime }\ }^{\ a}$ is inverse to $e_{\
a}^{a^{\prime }}.$ Let us consider that $\ ^{\circ }g_{i^{\prime }j^{\prime
}}=diag[\ ^{\circ }g_{1^{\prime }},\ ^{\circ }g_{2^{\prime }}],$ $%
h_{a^{\prime }b^{\prime }}=diag[\ ^{\circ }h_{3^{\prime }},\ ^{\circ
}h_{4^{\prime }}]$ and $\ ^{\circ }N_{i^{\prime }}^{a^{\prime }}=\{\ ^{\circ
}N_{i^{\prime }}^{3^{\prime }}=w_{i^{\prime }},\ ^{\circ }N_{i^{\prime
}}^{4^{\prime }}=n_{i^{\prime }}\}$ and (pseudo) Finsler data are $\ f_{ij},$
$\ f_{ab}$ and $\ ^{c}N_{i}^{a}$ $=\{\ ^{c}N_{i}^{3}=\ ^{c}w_{i},\
^{c}N_{i}^{4}=\ ^{c}n_{i}\},$ being parameterized by diagonal matrices, $\
f_{ij}=diag[f_{1},f_{2}]$ and $\ f_{ab}=diag[f_{3},f_{4}],$ if the
generating function is of type $F=\ ^{1}F(x^{i},y^{3})$ $+\
^{2}F(x^{i},y^{4})$ \ for some homogeneous (respectively, on $y^{3}$ and $%
y^{4})$ functions $\ ^{1}F$ and $\ ^{2}F.$
We may use arbitrary generating functions $F(x^{i},y^{a})$ but this will
result in off--diagonal (pseudo) Finsler metrics in N--adapted bases, which
would request a more cumbersome matrix calculus.

For simplicity, we can fix such nonholonomic distributions (fixing
correspondingly some generating/integration functions etc) when the
conditions (\ref{auxeq1}) are satisfied for a diagonal representation for $%
\mathbf{e}_{\ \alpha }^{\alpha ^{\prime }},$
\begin{equation*}
e_{\ 1}^{1^{\prime }}=\pm \sqrt{\left| \frac{\ f_{1}}{\ ^{\circ
}g_{1^{\prime }}}\right| },e_{\ 2}^{2^{\prime }}=\pm \sqrt{\left| \frac{\
f_{2}}{\ ^{\circ }g_{2^{\prime }}}\right| }\ ,e_{\ 3}^{3^{\prime }}=\pm
\sqrt{\left| \frac{\ f_{3}}{\ ^{\circ }h_{3^{\prime }}}\right| },e_{\
4}^{4^{\prime }}=\pm \sqrt{\left| \frac{\ f_{4}}{\ ^{\circ }h_{4^{\prime }}}%
\right| }.
\end{equation*}%
For any chosen values $\ f_{i},\ f_{a}$ and $\ ^{c}w_{i},^{c}n_{i}$ and
given $\ ^{\circ }g_{i^{\prime }}$ and $\ ^{\circ }h_{a^{\prime }},$ we can
compute $\ ^{\circ }w_{i^{\prime }}$ and $\ \ ^{\circ }n_{i^{\prime }}$ as
\begin{eqnarray*}
\ ^{\circ }w_{1^{\prime }} &=&\pm \sqrt{\left| \frac{\ ^{\circ }g_{1^{\prime
}}\ f_{3}}{\ ^{\circ }h_{3^{\prime }}\ f_{1}}\right| }\ ^{c}w_{1},\ \
^{\circ }w_{2^{\prime }}=\pm \sqrt{\left| \frac{\ ^{\circ }g_{2^{\prime }}\
f_{3}}{\ ^{\circ }h_{3^{\prime }}\ f_{2}}\right| }\ ^{c}w_{2}, \\
\ ^{\circ }n_{1^{\prime }} &=&\pm \sqrt{\left| \frac{\ ^{\circ }g_{1^{\prime
}}\ f_{4}}{\ ^{\circ }h_{4^{\prime }}\ f_{1}}\right| }\ ^{c}n_{1},\ \
^{\circ }n_{2^{\prime }}=\pm \sqrt{\left| \frac{\ ^{\circ }g_{2^{\prime }}\
f_{4}}{\ ^{\circ }h_{4^{\prime }}\ f_{2}}\right| }\ ^{c}n_{2},
\end{eqnarray*}%
corresponding to solutions of equations (\ref{auxeq2}).

\begin{corollary}
Any class of exact solutions of Einstein equations in general relativity
and/or Einstein--Finsler gravity, depending on corresponding sets of
generating/integration functions, defines a respective class of canonical
almost K\"{a}hler structures.
\end{corollary}

\begin{proof}
It is a result of Theorems \ref{thmr2}, \ref{th3b} and \ref{th6}. We have
classes of equivalence for data (an explicit example is given by formulas (%
\ref{auxeq1}) and (\ref{auxeq2}))
\begin{eqnarray*}
\left( \mathbf{g}_{\alpha \beta },N_{i}^{a},\mathbf{D}_{\alpha }\right)
&\sim &\left( \ ^{\circ }\mathbf{g}_{\alpha \beta },\ ^{\circ }N_{i}^{a},\
^{\circ }\widehat{\mathbf{D}}_{\alpha }\right) \sim \\
\left( \ \mathbf{f}_{\alpha \beta },\ ^{c}N_{i}^{a},\ ^{c}\widehat{\mathbf{D}%
}_{\alpha }\right) &\sim &\left( \ \ ^{\circ }\mathbf{\theta (\cdot ,\cdot )}%
\doteqdot \ ^{\circ }\mathbf{g}\left( \mathbf{J\cdot ,\cdot }\right) ,\ \
^{\circ }\mathbf{J,}\ ^{c}\widehat{\mathbf{D}}\equiv \ _{\theta }\widehat{%
\mathbf{D}}\right) .
\end{eqnarray*}
For Einstein manifolds, in general relativity, i.e. to encode data with the
Levi--Civita connection $\nabla ,$ we have to consider additional
constraints of type (\ref{lccond}), when $\ ^{c}\widehat{\mathbf{D}}%
\rightarrow \nabla .$ $\square $
\end{proof}

\vskip5pt

The above mentioned classes of exact solutions expressed in various forms
with nonholonomic/ Finsler / almost K\"{a}hler etc variables depend on the
type of generating/integration functions we chose, and what type of, for
instance, group/topological etc symmetries we prescribe for our geometric
and/or physical models. In a series of our works, see reviews of results in %
\cite{vrflg}, we constructed a number of examples with
Finsler like, and non--Finsler, black ellipsoid, Taub NUT, solitonic,
noncommutative, fractional etc gravitational solutions. It is an important
task to elaborate certain criteria for algebraic classifications of such
families of solutions; the geometric constructions should not depend
explicitly on the type of generating/integration functions.

\subsubsection{Algebraic index classification of Einstein and Finsler spaces}

In addition to Petrov's algebraic classification of Riemannian and Weyl
curvatures \cite{petrov} and related gravitational field configurations in
general relativity, we may provide a different algebraic classification of
gravitational fields, using Atiyah--Singer theorem for nonholonomic almost K%
\"{a}hler manifolds. Such an index classification is related to Fedosov
deformation quantization and can be performed for standard Einstein fields
and modifications.

\begin{claim}
Two solutions, $(\ _{1}^{\circ }\mathbf{g,}\ _{1}^{c}\widehat{\mathbf{D}})$
and $(\ _{2}^{\circ }\mathbf{g,}\ _{2}^{c}\widehat{\mathbf{D}}),$ of
Einstein equations (\ref{efspaces}) inducing two different nonholonomic
almost K\"{a}hler structures, $(\ _{1}^{\circ }\mathbf{\theta },\ _{\theta
}^{1}\widehat{\mathbf{D}})$ and $(\ _{2}^{\circ }\mathbf{\theta },\ _{\theta
}^{2}\widehat{\mathbf{D}}),$ i.e. two different $\ _{1}^{\mathcal{L}}\mathbf{%
K}^{2n}$ and $\ _{2}^{\mathcal{L}}\mathbf{K}^{2n},$ are noholonomically
equivalent and characterized by the same model of Fedosov quantization if
such solutions can be related via nonholonomic frame transforms, $\
_{1}^{\circ }\mathbf{g}_{\alpha \beta }=\ \mathbf{e}_{\ \alpha }^{\alpha
^{\prime }}\ \mathbf{e}_{\ \beta }^{\beta ^{\prime }}\ _{2}^{\circ }\mathbf{g%
}_{\alpha ^{\prime }\beta ^{\prime }},$ and they have the same cohomology classes
\begin{equation*}
cl(\ _{1}^{0}k)=cl(\ _{2}^{0}k),
\end{equation*}%
for any elements $\ _{1}^{0}k\in K_{0}(\ _{c}^{+}\mathcal{A}_{1})$ and $\
_{2}^{0}k\in K_{0}(\ _{c}^{+}\mathcal{A}_{2})$ (respectively, for $\ _{c}^{+}%
\mathcal{A}_{1}\mathcal{=}C(\ _{1}^{\mathcal{L}}\mathbf{K}^{2n})[[v]]$ and $%
\ _{c}^{+}\mathcal{A}_{2}\mathcal{=}C(\ _{2}^{\mathcal{L}}\mathbf{K}%
^{2n})[[v]]$ $),$ when
\begin{eqnarray*}
cl(\ _{1}^{0}k) &=&\left[ \widehat{A}(\ _{1}^{\mathcal{L}}\mathbf{K}^{2n})\
\exp \left( \frac{\ _{1}^{W}\widehat{\mathcal{C}}}{v}\right) \ ch\left(
\mathcal{K}(\ _{1}^{0}k)\right) \right] _{n+n}, \\
cl(\ _{2}^{0}k) &=&\left[ \widehat{A}(\ _{2}^{\mathcal{L}}\mathbf{K}^{2n})\
\exp \left( \frac{\ _{2}^{W}\widehat{\mathcal{C}}}{v}\right) \ \ ch\left(
\mathcal{K}(\ _{2}^{0}k)\right) \right] _{n+n},
\end{eqnarray*}%
$\ _{1}^{W}\widehat{\mathcal{C}}=-\ _{1}\mathbf{\theta }+\ _{1}^{v}\Omega $
and $\ _{2}^{W}\widehat{\mathcal{C}}=-\ _{2}\mathbf{\theta }+\
_{2}^{v}\Omega .$
\end{claim}

Let us provide two important motivations for such a claim.  Different
classes of exact solutions in classical gravity are defined by different
generating/integration functions and associated nonholonomic structures.
Under general frame/coordinate transforms, the parameterizations for
fundamental geometric objects change substantially. Such generic nonlinear
gravitational systems can be characterized topologically via corresponding
elliptic operators and their cohomology classes. This is also  an explicit
application of the algebraic index theorem in quantum gravity which allows
us to decide if two quantizations (in a generalized Fedosov sense) of some
nonholonomic gravitational configurations possess the same cohomological
characteristics, or not.

\appendix

\setcounter{equation}{0} \renewcommand{\theequation}
{A.\arabic{equation}} \setcounter{subsection}{0}
\renewcommand{\thesubsection}
{A.\arabic{subsection}}

\section{Einstein--Finsler Gravity in Almost Symplectic Variables}

\label{asect1}

Gravitational field equations in Einstein gravity on a (pseudo) Riemannian $%
\mathbf{V}$, and for Finsler gravity on $TM$, can be written
equivalently in terms of the Levi--Civita connection
$\nabla $, and using the almost symplectic connection
$\ _{\theta }\widehat{\mathbf{D}}$, both completely
defined by  the same fundamental geometric objects.
We summarize necessary formulas from \cite{vegpla,vlfedq,vrflg}.
We use the term Einstein--Finsler gravity for two different classes of gravity theories: the first one is for the usual general relativity written equivalently in Finsler variables and the second one is for Finsler gravity models on tangent bundles enabled with metric compatible d--connections.

\subsection{Torsion and curvature of normal d--connection}

Any d--connection $\mathbf{D}$ is characterized respectively by its torsion
and curvature tensors,
\begin{eqnarray}
\mathbf{T}(\mathbf{X},\mathbf{Y}) &\doteqdot &\mathbf{D}_{\mathbf{X}}\mathbf{%
Y}-\mathbf{D}_{\mathbf{Y}}\mathbf{X}-[\mathbf{X},\mathbf{Y}],  \label{ators}
\\
\mathbf{R}(\mathbf{X},\mathbf{Y})\mathbf{Z} &\doteqdot &\mathbf{D}_{\mathbf{X%
}}\mathbf{D}_{\mathbf{Y}}\mathbf{Z}-\mathbf{D}_{\mathbf{Y}}\mathbf{D}_{%
\mathbf{X}}\mathbf{Z}-\mathbf{D}_{[\mathbf{X},\mathbf{Y}]}\mathbf{Z},
\label{acurv}
\end{eqnarray}%
where $[\mathbf{X},\mathbf{Y}]\doteqdot \mathbf{XY}-\mathbf{YX,}$ for any
vectors $\mathbf{X}$ and $\mathbf{Y}.$

For the normal/almost symplectic d--connection $\widehat{\mathbf{D}}=\
_{\theta }\widehat{\mathbf{D}}=\{\widehat{\mathbf{\Gamma }}_{\beta \gamma
}^{\alpha }\}$ (\ref{cdcc}), we can consider the 1--form $\widehat{\mathbf{%
\Gamma }}_{j}^{i}=\widehat{L}_{jk}^{i}e^{k}+\widehat{C}_{jk}^{i}\mathbf{e}%
^{k},$ where $e^{k}=dx^{k}$ and $\mathbf{e}^{k}=dy^{k}+N_{i}^{k}dx^{k},$ we
can prove that the Cartan structure equations are satisfied,%
\begin{equation}
de^{k}-e^{j}\wedge \widehat{\mathbf{\Gamma }}_{j}^{k}=-\widehat{\mathcal{T}}%
^{i},\ d\mathbf{e}^{k}-\mathbf{e}^{j}\wedge \widehat{\mathbf{\Gamma }}%
_{j}^{k}=-\ ^{v}\widehat{\mathcal{T}}^{i},  \label{cart1}
\end{equation}%
and
\begin{equation}
d\widehat{\mathbf{\Gamma }}_{j}^{i}-\widehat{\mathbf{\Gamma }}_{j}^{h}\wedge
\widehat{\mathbf{\Gamma }}_{h}^{i}=-\widehat{\mathcal{R}}_{\ j}^{i}.
\label{cart2}
\end{equation}

The torsion 2--form $\widehat{\mathcal{T}}^{\alpha }=(\widehat{%
\mathcal{T}}^{i},\ ^{v}\widehat{\mathcal{T}}^{i}) =\widehat{\mathbf{T}}%
_{\ \tau \beta }^{\alpha }\ \mathbf{e}^{\tau }\wedge \mathbf{e}^{\beta }$ in
(\ref{cart1}) is computed:
\begin{equation*}
\widehat{\mathcal{T}}^{i}=\widehat{C}_{jk}^{i}e^{j}\wedge \mathbf{e}^{k},\
^{v}\widehat{\mathcal{T}}^{i}=\frac{1}{2}\Omega _{kj}^{i}e^{k}\wedge e^{j}+(%
\frac{\partial N_{k}^{i}}{\partial y^{j}}-\widehat{L}_{\ kj}^{i})e^{k}\wedge
\mathbf{e}^{j}.  \label{tform}
\end{equation*}%
i.e. the coefficients of torsion $\widehat{\mathbf{T}}_{\beta \gamma
}^{\alpha }$ (\ref{ators}) are
\begin{equation}
\widehat{T}_{jk}^{i}=0,\widehat{T}_{jc}^{i}=\widehat{C}_{\ jc}^{i},\widehat{T%
}_{ij}^{a}=\Omega _{ij}^{a},\widehat{T}_{ib}^{a}=e_{b}N_{i}^{a}-\widehat{L}%
_{\ bi}^{a},\widehat{T}_{bc}^{a}=0.  \label{cdtors}
\end{equation}%
It should be noted that $\widehat{\mathbf{T}}$ vanishes on h- and
v--subspaces, i.e. $\widehat{T}_{jk}^{i}=0$ and $\widehat{T}_{bc}^{a}=0,$
and the nontrivial h--v--components are induced nonholonomically and defined
canonically by component $\mathbf{g}$ and $\mathcal{L}.$

The curvature 2--form from (\ref{cart2}) of $\ \widehat{\mathbf{\Gamma }}%
_{\beta \gamma }^{\alpha }$ is computed%
\begin{equation}
\widehat{\mathcal{R}}_{\ \gamma }^{\tau }=\widehat{\mathbf{R}}_{\ \gamma
\alpha \beta }^{\tau }\ \mathbf{e}^{\alpha }\wedge \ \mathbf{e}^{\beta }=%
\frac{1}{2}\widehat{R}_{\ jkh}^{i}e^{k}\wedge e^{h}+\widehat{P}_{\
jka}^{i}e^{k}\wedge \mathbf{e}^{a}+\frac{1}{2}\ \widehat{S}_{\ jcd}^{i}%
\mathbf{e}^{c}\wedge \mathbf{e}^{d},  \label{cform}
\end{equation}%
when the nontrivial N--adapted coefficients of curvature $\ \widehat{\mathbf{%
R}}_{\ \beta \gamma \tau }^{\alpha }$ (\ref{acurv}) are
\begin{eqnarray}
\widehat{R}_{\ hjk}^{i} &=&\mathbf{e}_{k}\widehat{L}_{\ hj}^{i}-\mathbf{e}%
_{j}\widehat{L}_{\ hk}^{i}+\widehat{L}_{\ hj}^{m}\widehat{L}_{\ mk}^{i}-%
\widehat{L}_{\ hk}^{m}\widehat{L}_{\ mj}^{i}-\widehat{C}_{\ ha}^{i}\Omega
_{\ kj}^{a},  \label{cdcurv} \\
\widehat{P}_{\ jka}^{i} &=&e_{a}\widehat{L}_{\ jk}^{i}-\widehat{\mathbf{D}}%
_{k}\widehat{C}_{\ ja}^{i},\ \widehat{S}_{\ bcd}^{a}=e_{d}\widehat{C}_{\
bc}^{a}-e_{c}\widehat{C}_{\ bd}^{a}+\widehat{C}_{\ bc}^{e}\widehat{C}_{\
ed}^{a}-\widehat{C}_{\ bd}^{e}\widehat{C}_{\ ec}^{a}.  \notag
\end{eqnarray}

The N--adapted coefficients of the normal d--connection $\widehat{\mathbf{D}}%
=\ _{\theta }\widehat{\mathbf{D}}=\{\widehat{\mathbf{\Gamma }}_{\beta \gamma
}^{\alpha }\}$ and of the the Levi--Civita connection \ $\nabla =\{\
_{\shortmid }\Gamma _{\ \alpha \beta }^{\gamma }\}$ are related via formulas
\begin{equation}
\ _{\shortmid }\Gamma _{\ \alpha \beta }^{\gamma }=\widehat{\mathbf{\Gamma }}%
_{\ \alpha \beta }^{\gamma }+\ _{\shortmid }Z_{\ \alpha \beta }^{\gamma },
\label{cdeft}
\end{equation}%
where the distortion d--tensor $\ _{\shortmid }Z_{\ \alpha \beta }^{\gamma }$
is computed
\begin{eqnarray}
\ _{\shortmid }Z_{jk}^{a} &=&-\widehat{C}_{jb}^{i}g_{ik}g^{ab}-\frac{1}{2}%
\Omega _{jk}^{a},~_{\shortmid }Z_{bk}^{i}=\frac{1}{2}\Omega
_{jk}^{c}g_{cb}g^{ji}-\Xi _{jk}^{ih}~\widehat{C}_{hb}^{j},  \notag \\
\ _{\shortmid }Z_{jk}^{i} &=&0,\ _{\shortmid }Z_{bk}^{a}=~^{+}\Xi
_{cd}^{ab}~~\ \widehat{T}_{bk}^{c},\ _{\shortmid }Z_{kb}^{i}=\frac{1}{2}%
\Omega _{jk}^{a}g_{cb}g^{ji}+\Xi _{jk}^{ih}~\widehat{C}_{hb}^{j},
\label{cdeftc} \\
\ _{\shortmid }Z_{jb}^{a} &=&-~^{-}\Xi _{cb}^{ad}~~\widehat{T}_{dj}^{c},\
_{\shortmid }Z_{bc}^{a}=0,\ _{\shortmid }Z_{ab}^{i}=-\frac{g^{ij}}{2}\left[
\ \widehat{T}_{aj}^{c}g_{cb}+\ \widehat{T}_{bj}^{c}g_{ca}\right] ,  \notag
\end{eqnarray}%
for $e_{b}=\partial /\partial y^{a}$ and $\Xi _{jk}^{ih}=\frac{1}{2}(\delta
_{j}^{i}\delta _{k}^{h}-g_{jk}g^{ih}),~^{\pm }\Xi _{cd}^{ab}=\frac{1}{2}%
(\delta _{c}^{a}\delta _{d}^{b}\pm g_{cd}g^{ab}).$ The values (\ref{cdeft})
and (\ref{cdeftc}), and the h-- and v--components of $\ \widehat{\mathbf{%
\Gamma }}_{\ \beta \gamma }^{\alpha }$ given by (\ref{cdcc}) are determined
by coefficients of metric $\mathbf{g}$ on $\mathbf{V},$ and $\mathbf{N,}$ \
for a \ prescribed nonholonomic distribution with associated N--connection
structure.

\subsection{The Einstein equations for almost symplectic d--connec\-ti\-ons}

The Ricci tensor $\widehat{R}ic=\{\widehat{\mathbf{R}}_{\alpha \beta }\}$ of
$\ \widehat{\mathbf{D}}=\ _{\theta }\widehat{\mathbf{D}}$ can be defined in
standard form by contracting respectively the components of (\ref{cdcurv}), $%
\widehat{\mathbf{R}}_{\alpha \beta }\doteqdot \widehat{\mathbf{R}}_{\ \alpha
\beta \tau }^{\tau }.$ The scalar curvature is
\begin{equation}
\ ^{s}\widehat{\mathbf{R}}\doteqdot \mathbf{g}^{\alpha \beta }\widehat{%
\mathbf{R}}_{\alpha \beta }=g^{ij}\widehat{R}_{ij}+h^{ab}\widehat{R}_{ab},
\label{sdccurv}
\end{equation}%
where $\widehat{R}=g^{ij}\widehat{R}_{ij}$ and $\widehat{S}=h^{ab}\widehat{R}%
_{ab}$ are respectively the h-- and v--components of scalar curvature. This
allows a geometric formulation of the gravitational field equations for the
almost symplectic connection,

\begin{equation}
\widehat{\mathbf{E}}_{\ \beta \delta }=\widehat{\mathbf{R}}_{\ \beta \delta
}-\frac{1}{2}\mathbf{g}_{\beta \delta }\ ^{s}R=\widehat{\mathbf{\Upsilon }}%
_{\beta \delta }.  \label{ensteqcdc}
\end{equation}%
We can state well defined conditions when (\ref{ensteqcdc}) can be
constructed to be equivalent to the Einstein equations for $\nabla .$ This
is possible if $\widehat{\mathbf{\Upsilon }}_{\beta \delta }=\ ^{m}\mathbf{%
\Upsilon }_{\beta \delta }+\ ^{z}\mathbf{\Upsilon }_{\beta \delta }$ are
derived in such a way that they contain contributions from\ 1) \ the
N--adapted energy--momentum tensor $\ ^{m}\mathbf{\Upsilon }_{\beta \delta }$
(defined variationally following the same principles as in general
relativity but on $\mathbf{V}$) and 2), the distortion of the Einstein tensor
in terms of $\ \widehat{\mathbf{Z}}$ (\ref{cdeftc}), $\widehat{\mathbf{E}}%
_{\ \beta \delta }=\ _{\shortmid }E_{\alpha \beta }+\ ^{z}\widehat{\mathbf{E}%
}_{\ \beta \delta },$ for $\ ^{z}\widehat{\mathbf{E}}_{\ \beta \delta }=\
^{z}\mathbf{\Upsilon }_{\beta \delta }.$\footnote{%
The value $\ ^{z}\widehat{\mathbf{E}}_{\ \beta \delta }$ is computed by
introducing $\widehat{\mathbf{D}}=\nabla -\widehat{\mathbf{Z}}$ into (\ref%
{ensteqcdc}) and corresponding contractions of indices in order to find the
Ricci d--tensor and scalar curvature.}

The equations (\ref{ensteqcdc}) are considered as the fundamental field
equations in Einstein--Finsler gravity (the d--connection $\ _{\theta }%
\widehat{\mathbf{D}}$ is also a Finsler/Lagrange connection), see details in
Refs. \cite{vrflg}. They transform into usual Einstein equations
in general relativity if

\begin{equation}
\widehat{L}_{aj}^{c}=e_{a}(N_{j}^{c}),\ \widehat{C}_{jb}^{i}=0,\ \Omega _{\
ji}^{a}=0,  \label{lccond}
\end{equation}%
for $\mathbf{\Upsilon }_{\beta \delta }\rightarrow \varkappa T_{\beta \delta
}$ (matter energy--momentum in Einstein gravity)\ if $\ \widehat{\mathbf{D}}%
\rightarrow \nabla .$

\setcounter{equation}{0} \renewcommand{\theequation}
{B.\arabic{equation}} \setcounter{subsection}{0}
\renewcommand{\thesubsection}
{B.\arabic{subsection}}

\section{Proof of Main Theorem}

\label{asproof}

\subsection{A technical Lemma}

Let us prove an obvious ''nonholonomic'' analogue of Lemma 1 in \cite{chen}
which have technical importance for proving the Main Result of this paper.

\begin{lemma}
Let us consider a nonholonomic vector bundle $\mathcal{E}$ $\ $of rank $k$
over $\ ^{\mathcal{L}}\mathbf{K}^{2n}$ endowed with a Fedosov--Finsler
normal d--connection $\ _{r}^{d}\mathcal{D}=\ ^{d}\mathcal{D}+v^{-1}\left[ \
^{\theta }r,\cdot \right] $ (\ref{ffdc}) when $\delta ^{-1}\ {}^{\theta }r=0$,
following the conditions of Corollary \ref{corolrem}. For any endomorphism
of $\mathcal{E},$ $\varsigma \in End_{V},$ and connection $^{\ \mathcal{E}%
}\partial $, if $^{\ \mathcal{E}}\partial \varsigma =0$ we get $\ _{r}^{d}%
\mathcal{D}\varsigma =0,$ i.e. the isomorphism $\chi $ (\ref{isom})
transforms $\varsigma $ into itself.
\end{lemma}

\begin{proof}
If $\varsigma $ does not depend on $z$--variables, i.e. $^{\ \mathcal{E}%
}\partial \varsigma =0,$ we have $\ ^{d}\mathcal{D}\varsigma =\delta
\varsigma =0$ and $\ _{r}^{d}\mathcal{D}\varsigma =v^{-1}$ $\left[ \
^{\theta }r,\varsigma \right] $ when (for a nilpotent \ $\ _{r}^{d}\mathcal{D%
}\mathbf{)}$ \ we must have $\ _{r}^{d}\mathcal{D}\left[ \ ^{\theta
}r,\varsigma \right] =0.$ For our purposes, we should prove that $\left[ \
^{\theta }r,\varsigma \right] $ vanishes.

Since $\delta ^{-1}\ {}^{\theta }r=0$ and $\varsigma $ does not depend on $z$%
--variables,
$\delta ^{-1}\left[\ ^{\theta }r,\varsigma \right] =0.$ We compute $\delta %
\left[\ ^{\theta }r,\varsigma \right] =\ ^{d}\mathcal{D}\left[\ ^{\theta
}r,\varsigma \right] +v^{-1}\left[ \ ^{\theta }r,\left[ \ ^{\theta
}r,\varsigma \right] \right],$ when $\left[ \ ^{\theta }r,\varsigma \right]
\in \ ^{d}\Omega ^{1}(\mathbf{W}(End_{V})).$ Applying  (\ref{aux11})
to $\left[\ ^{\theta }r,\varsigma \right]$ we get the equation%
 $\left[\ ^{\theta }r,\varsigma \right] =\delta ^{-1}\left(\ ^{d}\mathcal{D}%
\left[\ ^{\theta }r,\varsigma \right] +v^{-1}\left[\ ^{\theta }r,\varsigma %
\right] \right).$
There is only one non--contradictory solution $\left[\ ^{\theta
}r,\varsigma \right] =0$ because $\delta ^{-1}$ increases the degree in $z$%
--variables, but such a commutator does not. $\square $
\end{proof}

\vskip5pt

We emphasize that on $\ ^{\mathcal{L}}\mathbf{K}^{2n}$ endowed with certain
canonical symplectic forms and N--adapted connections we can work similarly
as on K\"{a}hler manifolds but keeping the constructions to be distinguished
nonholonomically as some h-- and v--components of a corresponding almost K%
\"{a}hler geometry.

\subsection{Sketch of proof for theorem \ref{th5}}

The method we should apply is inspired from \cite{chen}, see there the
Theorem 4. The idea is to prove that distinguished $(n+n)$--form $\ ^{%
\mathcal{D}}\digamma \left( \ ^{p}\zeta _{1}\right) $ has the same
cohomology class as the $(n+n)$--th component of the form \
\begin{equation*}
\det \left( \frac{\ ^{s}\widehat{\mathbf{R}}/2}{\sinh (\ ^{s}\widehat{%
\mathbf{R}}/2)}\right) ^{1/2}\ e^{v^{-1}(\theta -\ ^{v}\Omega )}tr\ e^{\ ^{V}%
\mathcal{R}},
\end{equation*}%
where $\ ^{s}\widehat{\mathbf{R}}$ (\ref{sdccurv}) is the scalar curvature
of the normal d--connection, $\ ^{W}\widehat{\mathcal{C}}=-\theta +\
^{v}\Omega ,$ see (\ref{fwdc}), $^{V}\mathcal{R}$ is the scalar curvature
form on the d--vector bundle $\mathcal{E}$ and $tr$ is the ordinary trace of
matrices. Such a formula is similar to that derived in Theorem \ref{th1a}.

The above Lemma is necessary for evaluating principal parts and idempotents
of involved d--algebras. Computations are performed similarly for the h--
and v--components dubbing the formulas provided in the holonomic K\"{a}hler
constructions. We leave the technical details to the interested reader.


\begin{thebibliography}{99}
\bibitem{fn01} a pair $\left( \mathbf{V},
\mathcal{N}\right) $ consisting of a manifold $\mathbf{V}$ and a
non--integrable distribution $\mathcal{N}$ on this space is called a
nonholonomic manifold
\bibitem{fn02}
For different purposes, we can consider such a space to be a
semi--Riemannian/ Riemann--Cartan manifold, a real/ complex (co) vector and
tangent bundle, a Clifford module; we can extend the approach to various
noncommutative and supersymmetric generalizations etc.
\bibitem{fn03}
physicists use instead of "semi" the term "pseudo", for local
pseudo--Euclidean signatures of type (-,+,+,+);\ in this paper, we consider  that a pseudo--Riemannian spaces is a manifold of necessary smooth class and endowed with a metric structure of local pseudo--Euclidean signature
\bibitem{fn04}
we shall use "$:=$" instead of "by definition"; here we also note that a symbol like $\mathbf{e}_{i}$ is boldface because it contains N--elongated terms;  $e_b$ is not boldface because it is just a partial derivative and do not include N--coefficients; for  dual values, we write $e^i$ but $\mathbf{e}^b$ because only in the second case there are N--elongations
\end{thebibliography}
\end{document}